\documentclass[11pt,letterpaper]{article}
\usepackage{latexsym,color,amssymb,amsthm,fullpage,enumitem,
amsmath,amsfonts,bm}

\usepackage{graphicx}
\usepackage{euscript}
\usepackage{graphicx}

\usepackage[left=1in, right=1in, top=1in, bottom=1in]{geometry}

\newcommand{\ignore}[1]{}

\newtheorem{theorem}{Theorem}[section]
\newtheorem{lemma}[theorem]{Lemma}

\newtheorem{proposition}[theorem]{Proposition}

\newtheorem{corollary}[theorem]{Corollary}

\def\eps{{\epsilon}}

\def\conv{{{\mathtt{conv}}}}

\def\1{{(1)}}
\def\2{{(2)}}

\def\deg{{\mathtt{deg}}}

\def\A{{\cal A}}

\def \sign {{\sf sign}}
\def\U{{\cal U}}

\def\F{{\cal F}}

\def\W{{\cal W}}

\def\N{{\cal N}}
\def\M{{\cal M}}
\def\SS{{\mathbb S}}
\def\reals{{\mathbb R}}

\def\P{{\EuScript{P}}}

\def\V{{\cal V}}

\def\G{{\cal G}}

\def\T{{{\EuScript{T}}}}

\begin{document}

\begin{titlepage}

\title{An Efficient Regularity Lemma for Semi-Algebraic Hypergraphs}

\author{
Natan Rubin\thanks{Email: {\tt rubinnat.ac@gmail.com}. Ben Gurion University of the Negev, Beer-Sheba, Israel. Supported
by grant 2891/21 from Israel Science Foundation.}
}

%\date{}
\maketitle

\begin{abstract}
The vast majority of hypergraphs that arise in discrete and computational geometry, describe {\it semi-algebraic relations} between elementary geometric objects.
We use the polynomial method of Guth and Katz to establish stronger and {\it more efficient} regularity and density theorems for such $k$-uniform hypergraphs $H=(P,E)$,
 where $P$ is a finite point set in $\reals^d$, and the edge set $E$ is determined by a semi-algebraic relation of bounded description complexity.

In particular, for any $0<\eps\leq 1$ we show that one can construct in $O\left(n\log 1/\eps\right)$ time, an equitable partition $P=U_1\uplus \ldots\uplus U_K$ into $K=O(1/\eps^{d+1+\delta})$ subsets, for any $0<\delta$, so that all but $\eps$-fraction of the $k$-tuples $U_{i_1},\ldots,U_{i_k}$ are {\it homogeneous}:
 we have that either $U_{i_1}\times\ldots\times U_{i_k}\subseteq E$ or $(U_{i_1}\times\ldots\times U_{i_k})\cap E=\emptyset$. If the points of $P$ can be perturbed in a general position, the bound improves to $O(1/\eps^{d+1})$, and the partition is attained via a {\it single partitioning polynomial} (albeit, at expense of a possible increase in the worst-case running time). 
 
 The best previously known such partition, due to Fox, Pach and Suk requires $\Omega\left(n^{k-1}/\eps^{c}\right)$ time and yields $K=1/\eps^c$ parts (for $0<\eps\leq 1/4$), where $c$ is an enormous constant which is not stated explicitly and depends not only on the dimension $d$ but also on the semi-algebraic description complexity of the hypergraph.

In contrast to the previous such regularity lemmas which were established by Fox, Gromov, Lafforgue, Naor, and Pach and, subsequently, Fox, Pach and Suk, our partition of $P$ does not depend on the edge set $E$, provided its semi-algebraic description complexity does not exceed a certain constant.  

 As a by-product, we show that in any $k$-partite $k$-uniform hypergraph $(P_1\uplus\ldots\uplus P_k,E)$ of bounded semi-algebraic description complexity in $\reals^d$ and with $|E|\geq \eps \prod_{i=1}^k|P_i|$ edges, one can find, in expected time $O\left(\sum_{i=1}^k\left(|P_i|+1/\eps\right)\log (1/\eps)\right)$, subsets $Q_i\subseteq P_i$ of cardinality $|Q_i|\geq |P_i|/\eps^{d+1+\delta}$, so that $Q_1\times\ldots\times Q_k\subseteq E$.

%Comparable results are obtained for asymmetric settings which involve $k$-partite hypergraphs $(V_1,\ldots,V_k,E)$, whose parts $V_i$ lie in spaces $\reals^{d_i}$ of possibly different dimensions $d_i$.  

%In particular, we show that, for any $\eps$-dense $k$-partite and $k$-uniform hypergraph $(V_1,\ldots,V_k,E)$ with bounded semi-algebraic description complexity in $\reals^d$, there exist hom

%We then apply the new regularity lemma to improve state-of-the-art for testing hereditary properties of semi-algebraic hypergraphs, and derive more general and algorithmically efficient analogues of the same-type lemma.

%In particular, any $n$-point set $P$ in general position in $\reals^d$ can be subdivided, by the means of Matou\v{s}ek's simplicial partition, into $r$ pairwise disjoint parts $P_1,\ldots,P_r$, of size $\Theta\left(n/r\right)$ each, so that all but $O\left(r^{d+1-1/d}\right)$ of the $(d+1)$-tuples $P_{i_1},\ldots,P_{i_{d+1}}$, with $1\leq i_1<i_2<\ldots< i_k\leq r$, attain a fixed orientation. 

%These estimates constitute a drastic improvement over the so called polynomial regularity lemma of Fox, Pach, and Suk.

\end{abstract}

\maketitle

%\blfootnote{\textup{2010} \textit{Mathematics Subject Classification}: \textup{52A35, 52C10, 52C15, 52C30,  52C35, 52C45, 05D40}}

\end{titlepage}

\section{Introduction} \label{sec:intro}

\subsection{Semi-algebraic hypergraphs}

%An essential ingredient of several selection-type results (most notably, a recent bound for Pach's selection lemma \cite{PachTheorem}) is the existence of very strong regularity and Tur\'an-type results for $k$-uniform, $k$-partite hypergraphs of {\it bounded semi-algebraic description complexity} in $\reals^{d}$. 

 We say that a set $X\subseteq \reals^{d}$ is {\it semi-algebraic} if it is the locus of all the points in $\reals^{d}$ that satisfy a particular Boolean combination of a finite number of polynomial equalities and inequalities in the $d$ coordinates $x_1,\ldots,x_d$ of $\reals^{d}$ \cite{BasuBook}.
 The majority of problems in computational geometry happen to involve families of ``simply shaped" semi-algebraic sets in $\reals^d$ (e.g., lines, hyperplanes, axis-parallel boxes, or discs).
Hence, the basic relations between such semi-algebraic objects are typically captured by hypergraphs of {\it bounded semi-algebraic description complexity}. 

We say that a Boolean function $\psi:\reals^{d\times k}\rightarrow \{0,1\}$ is a {\it $k$-wise semi-algebraic relation in $\reals^{d}$} if it can be described by a finite combination $(f_1,\ldots,f_{s},\phi)$, where $f_1,\ldots,f_{s}\in \reals[x_1,\ldots,x_d]$ are $d$-variate polynomials, and $\phi$ is a Boolean function in $\{0,1\}^{s}\rightarrow \{0,1\}$, so that 
$\psi(y_1,\ldots,y_k)=\phi(f_1(y_1,\ldots,y_k)\leq 0;\ldots;f_{s}(y_1,\ldots,y_{k})\leq 0)$ holds for all the ordered $k$-tuples $(y_1,\ldots,y_k)$ of points $y_i\in \reals^d$.\footnote{In the sequel, such $k$-tuples $(y_1,\ldots,y_k)$ are routinely interpreted as elements of the $(dk)$-dimensional space $\reals^{d\times k}$.}
In what follows, we refer to $(f_1,\ldots,f_{s},\phi)$ as a {\it semi-algebraic description of $\psi$} (which need not be unique). 

We say that such a relation $\psi$ has {\it description complexity at most $(\Delta,s)$} if it admits a description using at most $s$ polynomials $f_i\in \reals[x_1,\ldots,x_d]$ of maximum degree at most $\Delta$.
We then use $\Psi_{d,k,\Delta,s}$ to denote the family of all such $k$-wise semi-algebraic relations $\psi:\reals^{d\times k}\rightarrow \{0,1\}$ whose respective description complexities are bounded by $(\Delta,s)$.

Let $P_1,\ldots,P_k$ be pairwise disjoint finite subsets of $\reals^d$. 
We say that a $k$-partite, $k$-uniform hypergraph $(P_1,\ldots,P_k,E)$ has {\it semi-algebraic description complexity at most $(\Delta,s)$ in $\reals^d$} if its edge set
$E$ is determined, {as a subset of $P_1\times\ldots\times P_k$}, by a $k$-wise relation $\psi\in \Psi_{d,k,\Delta,s}$; namely, a $k$-tuple $(p_1,\ldots,p_k)\in P_1\times\ldots\times P_k$ belongs to $E$ if and only if $\psi(p_1,\ldots,p_k)=1$. The definition naturally extends to more general classes of hypergraphs, which are induced by symmetric relations $\psi\in \Psi_{d,k,\Delta,s}$.

The systematic study of semi-algebraic {\it graphs} in Euclidean spaces was initiated by Alon, Pach, Pinchasi, Radoici\'c, and Sharir \cite{AlonSemi}, who established that any graph $G=(P,E)$ of bounded sem-algebraic description complexity in $\reals^d$ has the so called {\it strong Erd\H{o}s-Hajnal property}: either $G$ or its complement $\overline{G}$ must contain a copy of $K_{t,t}$, for some $t=\Theta(n)$. As a result, $G$ must contain, as an induced subgraph, either a clique or an independent set of size $n^{c}$, where $c$ is a constant the depends on the dimension $d$ and the semi-algebraic description complexity of $G$.

The related {\it Zarankiewicz problem} concerns the maximum number of edges in a $K_{t,t}$-free bi-partite graph (where $t$ is typically a fixed constant). 
If the graph $G=(P_1,P_2,E)$ in question is semi-algebraic of bounded description complexity in $\reals^d$, then a study by Fox, Pach, Sheffer, Suk and Zahl  \cite{Zarankiewicz} yields a generally sharp bound $|E|=O\left(\left(|P_1|\cdot |P_2|\right)^{\frac{d}{d+1}+\delta}+|P_1|+|P_2|\right)$, where $t$ affects the implicit constant of proportionality. This bound was subsequently generalized to $K_{t,t\ldots,t}$-free $k$-partite and $k$-uniform hypergraphs \cite{DoHyper}. Even better upper estimates are known for particular classes of semi-algebraic and other geometrically-defined graphs \cite{SemiLinear,ChaHar,FranklKupavskii,VC,ZaraNets}; also see \cite{Istvan} for related results that are known for so called algebraic hypergraphs over finite fields.

\subsection{Regularity and density theorems for semi-algebraic hypergraphs} 
 A partition of a finite set is called {\it equitable} if the cardinalities of any two parts differ by at most one. According to Szemer\'edi's regularity lemma \cite{Szemeredi}, which underlies modern extremal combinatorics, for every $\eps> 0$ there exists a finite number $K=K(\eps)$ such that every graph admits an equitable vertex partition into at most $K$ parts with the following property: all but at most
an $\eps$-fraction of the pairs of parts determine ``almost-random" bipartite sub-graphs. The original proof yields $K(\eps)$ that is an exponential tower of 2-s of height $\eps^{-O(1)}$, which was demonstrated by Gowers \cite{GowersLower} to be necessary for certain graphs. 
More recently, Szemer\'edi's lemma has been extended to $k$-uniform hypergraphs by Gowers \cite{GowersHyper,GowersLower2} and Nagle {\it et al.} \cite{Nagle}; as a rule, the partition size climbs up in the Ackermann hierarchy as $k$ increases.
Unfortunately, such bounds have detrimental impact on the efficiency of the property testing algorithms, and other applications of the regularity lemma.

%The past 15 years have seen a sequence of progressively stronger analogues of Szemeredi's regularity lemma for semi-algebraic hypergraphs. 

%To proceed, several definitions are in order. 

Specializing to the hypergraphs $H=(P,E)$ of bounded semi-algebraic description complexity in $\reals^d$, a far stronger regularity statement was obtained in 2010 by Fox, Gromov, Lafforgue, Naor and Pach \cite{Overlap}. 
To this end, we say that a $k$-tuple of pairwise disjoint subsets $Q_1,\ldots,Q_k\subseteq P$ is {\it homogeneous} in $H$ if we have that either $Q_1\times\ldots \times Q_k\subseteq E$ or  $\left(Q_1\times\ldots \times Q_k\right)\cap E=\emptyset$. 
% An equitable subdivision $P=U_1\uplus\ldots\uplus U_K$ of $V$ is {\it $\eps$-regular} if all by an $\eps$-fraction of the $k$-size sequences $(U_{i_1},\ldots,U_{i_k})$ (of distinct subsets) are homogeneous with respect to $E$.

\begin{theorem}[Fox, Gromov, Lafforgue, Naor, and Pach \cite{Overlap}]\label{Theorem:GromovRegularity}
For any $0<\varepsilon\leq 1$, and
any hypergraph $H=(P,E)$ that admits a semi-algebraic description of bounded complexity $(\Delta,s)$ in $\reals^d$, there is a subdivision $P=U_1\uplus \ldots \uplus U_K$ into $K=K(d,k,\Delta,s,\eps)$ parts with the following property: all by an $\eps$-fraction of the sequences $(U_{i_1},\ldots,U_{i_k})$, each comprised of $k$ distinct subsets $U_{i_j}$, are homogeneous with respect to $H$.
 %Here $c$ is a constant which depends only on $d$, $k$, and the fixed constants $D$ and $s$ that bound the semi-algebraic description complexity of $\G$.
\end{theorem}

At the heart of the original proof of Theorem \ref{Theorem:GromovRegularity} by Fox {\it et al.} lies a ``density theorem" of independent interest. 
To this end, we say that a $k$-uniform, $k$-partite hypergraph $(P_1,\ldots,P_k,E)$ is {\it $\eps$-dense} if we have that $|E|\geq \eps |P_1|\cdot\ldots\cdot |P_k|$. 

\begin{theorem}[\cite{Overlap}]\label{Theorem:GromovTuran}
For any fixed integers $d\geq 0$, $k\geq 0$, $\Delta\geq 0$, and $s>0$, there is a constant $c=c(d,k,\Delta,s)$ with the following property: For any $0<\eps\leq 1$, and
any $k$-uniform, $k$-partitite and $\eps$-dense hypergraph $H=(P_1,\ldots,P_k,E)$, that admits a semi-algebraic description of bounded complexity $(\Delta,s)$ in $\reals^d$, there exist subsets $Q_1\subseteq P_1,\ldots,Q_k\subseteq P_k$ which meet the following criteria
\begin{enumerate}
	\item  We have that $|Q_i|=\Omega\left(\eps^{c}|P_i|\right)$ for all $1\leq i\leq k$, where the implicit constant of proportionality depends on $k$, $d$, $\Delta$, and $s$.
	\item  $Q_1\times\ldots\times Q_k\subseteq E$.
\end{enumerate}
%Here $c$ is a constant which depends only on $d$, $k$, and the fixed constants $D$ and $s$ that bound the semi-algebraic description complexity of $\G$.
\end{theorem}

The subsequent {\it polynomial regularity lemma} of Fox, Pach and Suk \cite{Regularity} improves the maximum number of parts in Theorem \ref{Theorem:GromovRegularity} to $O\left(1/\eps^{c'}\right)$. To this end, they use Chazelle's Cutting Lemma \cite{Cuttings} in dimension ${\Delta+d\choose d}-1$ to establish a stronger bound for Theorem \ref{Theorem:GromovTuran}, with exponent $c={\Delta+d\choose d}$. Though the exponent $c'=c'(k,d,\Delta,s)$ is nowhere specified, it is clearly enormous and depends the description complexity of $H$; furthermore, the overall construction runs in time that is at least proportional to $|P|^{k-1}/\eps^{c'}$.

Proving sharp asymptotic estimates for Theorems \ref{Theorem:GromovRegularity} and \ref{Theorem:GromovTuran} would be instrumental for tackling a plethora of fundamental problems in computational and combinatorial geometry, including the $k$-set problem \cite{Rubin}, Ramsey and coloring numbers of semi-algebraic graphs and hypergraph \cite{Conlon,SchurErdos}, higher dimensional variants of Erd\H{o}s Happy Ending Problem \cite{HappyEndHigher}, and efficient testing of hereditary properties in geometrically defined hypergraphs \cite{Regularity}.
%At the heart of the proof of Theorem \ref{Theorem:FoPaSuRegularity} lies the following Tur\'an-type statement of independent interest:
%For any $\varepsilon$-dense, $k$-uniform, and $k$-partite hypergraph $(V_1,V_2,\ldots,V_k,E)$ of bounded semi-algebraic description complexity and, in particular, of bounded maximum degree $D$, it yields subsets $W_i\subseteq V_i$ of size $|W_i|=\Omega\left(\varepsilon^{d+D\choose d}|V_i|\right)$, and so that $W_1\times \ldots\times W_k\subseteq E$.

For example, in SODA 2024 the author \cite{Rubin} obtained
a sharper variant of Theorem \ref{Theorem:GromovTuran} in which the cardinality of each
subset $Q_i\subseteq P_i$ is roughly\footnote{As a matter of fact, the author's density theorem yields a slightly more favourable ``aggregate" estimate $|Q_1|\cdot\ldots\cdot |Q_k|=\Omega\left(\eps^{(k-1)d+1}|V_1|\cdot\ldots\cdot |V_k|\right)$.}
 $\eps^{d}|P_i|$. He then used the new bound to dramatically improve the lower bound for the point selection problem \cite{AlonSelections} together with the related upper bound for halving hyperplanes in dimension $d\geq 5$ \cite{BFL}. 

It is natural to ask if there also exists a more efficient variant of Theorem \ref{Theorem:GromovRegularity}, with $O\left((1/\eps)^{c_d}\right)$ parts, where $c_d$ is a (hopefully small) constant that {\it depends only on the dimension $d$}?
Unfortunately, plugging author's density theorem into the overall analysis of Fox {\it et al.} \cite{Regularity} does {\it not} yield a comparably efficient polynomial regularity lemma. Furthermore, the implied algorithm for {\it finding} the subsets $Q_i\subseteq P_i$, so that $Q_1\times\ldots\times Q_k\subseteq E$, still runs in time $n^{k-1}/\eps^{O(1)}$.

%Our proof of the new bound for the point selection problem will require a more efficient Tur\'an-type theorem for semi-algebraic hypergraphs, with no dependence on the maximum degree $D$ in the exponent of $\varepsilon$. As a by-product, we derive a more efficient variant of the polynomial regularity lemma.

\subsection{Main results}

\noindent{\bf Efficient regularity for semi-algebraic hypergraphs.} In this paper we establish the first {\it explicit} polynomial bound for Theorem \ref{Theorem:GromovRegularity}. To this end, we use the polynomial partition of Guth and Katz \cite{GuthKatz} to construct a partition with $K=O\left(1/\eps^{d+1+\delta}\right)$ parts, for any $\delta>0$.\footnote{As a rule, bounds of this form involve implicit multiplicative factors that may be exponential in $1/\delta$.} Unlike all the previous semi-algebraic regularity lemmas of Fox {\it et al} \cite{Overlap,Regularity}, our regular partition of the set $P$ is {\it oblivious to the edge set $E$} as long as the latter is determined by a $k$-wise relation $\psi\in \Psi_{d,k,\Delta,s}$. 
%For the rest of this paper the hypergraphs $(V,E)$ and $k$-wise relations $\psi:\reals^{d\times k}\rightarrow \{0,1\}$ that are not necessarily symmetric.

%To this end, a more comprehensive notation is in order.

For the sake of brevity, an ordered sequence $(P_1,\ldots,P_k)$ of $k$ pairwise disjoint subsets $P_i\subseteq \reals^d$ is called a {\it $k$-family} in $\reals^d$. 
Let $\psi:\reals^{d\times k}\rightarrow \{0,1\}$ be a $k$-wise relation over $\reals^d$.
A $k$-family $(U_1,\ldots,U_k)$ is {\it $\psi$-homogeneous} if the value $\psi\left(p_1,\ldots,p_k\right)$ is invariant in the choice of vertices $p_{1}\in U_{1},\ldots,p_k\in U_k$.
 
For any $k$-wise relation $\psi:\reals^{d\times k}\rightarrow\{0,1\}$, $\eps>0$, and any finite point set $P\subset \reals^d$,
we say that an equitable partition $P=U_1\uplus\ldots\uplus U_K$ is {\it $(\eps,\psi)$-regular} if all but $\eps$-fraction of the $k$-families $(U_{i_1},\ldots,U_{i_k})$ are $\psi$-homogeneous.

For any class $\Psi\subseteq \reals^{d\times k}\rightarrow \{0,1\}$ of $k$-wise relations over $\reals^d$, $\eps>0$, and any finite point set $P\subset \reals^d$, we say that an equitable partition $P=U_1\uplus\ldots\uplus U_K$ of $P$ is {\it $(\eps,\Psi)$-regular} if it is $(\eps,\psi)$-regular for {\it any} $k$-wise relation $\psi\in \Psi$. 

The main result of this paper is the existence of a small-size and {\it easily constructible} $(\eps,\Psi_{d,k,\Delta,s})$-regular partition for any point set $P$ in $\reals^d$.\footnote{Notice that if the equitable partition $P=U_1\uplus\ldots\uplus U_K$ is $(\eps,\Psi_{d,k,\Delta,s})$-regular then, {regardless of the implicit ordering of the vertices of $P$}, it is $(\eps\cdot k!)$-regular with respect to any $k$-uniform hypergraph $(P,E)$ whose semi-algebraic description complexity is bounded by $(\Delta,s)$.}

\begin{theorem}\label{Theorem:NewRegularity} 
For any fixed integers $d>0$, $k>1$, $\Delta\geq 0$, $s\geq 1$, and any $\delta>0$, there is a constant $c=c(d,k,\Delta,s,\delta)>0$ with the following property.

	Let $\eps>0$, and $P\subset \reals^d$ an $n$-point set. Then there exists an $\left(\eps,\Psi_{d,k,\Delta,s}\right)$-regular partition $P=U_1\uplus U_2\uplus\ldots\uplus U_K$ whose cardinality satisfies $K\leq c/\eps^{d+1+\delta}$. 
	\end{theorem}

A major obstacle to routine use of the partition theorem of Guth and Katz (see Theorem \ref{Thm:PolynomialPartition} in Section \ref{Subsec:Polynomial}) is that too many points of $P$ may fall in the zero set $Z(f)$ of the partitioning polynomial and, thereby, remained un-assigned to any cell of $\reals^{d}\setminus Z(f)$.
A veritable way to reduce the number of such points is to apply a random and independent {\it $\vartheta$-perturbation} to each point $p\in P$, with $\vartheta>0$, so that each coordinate $x_i(p)$ is independently ``shifted" by a random value $h_i(p)\in [-\vartheta,\vartheta]$.
	As a result, the proof of Theorem \ref{Theorem:NewRegularity} can be dramatically simplified, and its bound marginally improved, over a broad sub-class of ``generic" hypergraphs whose vertices can be perturbed in a general position \cite{GeneralPosition}.

To this end, we say that a $k$-wise relation $\psi:\reals^{d\times k}\rightarrow \{0,1\}$ is {\it sharp} over a point set $P$ if for any sequence $(p_1,\ldots,p_k)$ of $k$ distinct vertices $p_i\in P$, the value of $\psi(p_1,\ldots,p_k)$ remains invariant under any $\vartheta$-perturbation of $P$, for any sufficiently small $\vartheta>0$.

\begin{theorem}\label{Theorem:NewRegularitySharp} 
For any fixed integers $d>0$, $k>1$, $\Delta\geq 0$, $s\geq 1$, there is a constant $c=c(d,k,\Delta,s)>0$ with the following property.

	Let $\eps>0$, and $P\subset \reals^d$ be a finite point set. Then $P$ admits a partition $P=U_1\uplus U_2\uplus\ldots\uplus U_K$ of cardinality $K\leq c/\eps^{d+1}$ that is $(\eps,\psi)$-regular for any $k$-wise relation $\psi\in \Psi_{d,k,\Delta,s}$ that is sharp over $P$.
	\end{theorem}

	\noindent{\bf Algorithmic aspects.} In contrast to the previous proof of the polynomial regularity lemma, due to Fox, Pach and Suk \cite{Regularity}, our analysis yields a near-linear time algorithm for computing $\left(\eps,\Psi_{d,k,\Delta,s}\right)$-regular partitions.

	\begin{theorem}\label{Theorem:Construct}
	The $(\eps,\Psi_{d,k,\Delta,s})$-regular partition of Theorem \ref{Theorem:NewRegularity} can be constructed in time  in expected time $O\left(n\log(1/\eps)\right)$. Furthermore, the partition algorithms yields a data structure that, given a query relation $\psi\in \Psi_{d,k,\Delta,s}$, returns in additional $O\left(1/\eps^{(d+1)k-1+\delta}\right)$ time, a collection of $k$-families $(U_{i_1},\ldots,U_{i_k})$ which encompasses at most an $\eps$-fraction of all the $k$-families, including all such $k$-families that are not $\psi$-homogeneous.\footnote{Every $k$-family $(U_{i_1},\ldots,U_{i_k})$ in the output is represented by the respective sequence $(i_1,\ldots,i_k)$.}
	\end{theorem}
	
	It is an outstanding algorithmic problem to quickly decide whether a $k$-family $(U_1,\ldots,U_k)$ is $\psi$-homogeneous with respect to a given relation $\psi\in \Psi_{d,k,\Delta,s}$.\footnote{The particular 3-partite $1$-dimensional instance, where $\psi(x_1,x_2,x_3)=1$ if and only if $x_3=x_1+x_2$, has become notorious as the 3SUM problem \cite{3SUM}. It is not known whether the task can be performed in time $O((|U_1|+|U_2|+|U_3|)^{2-\gamma})$, for any $\gamma>0$.} Theorem \ref{Theorem:Construct} gets around this obstacle but returning a small-size and more robustly defined {\it superset}, which includes all $\psi$-homogeneous $k$-families.

	\medskip
	Our machinery yields the following constructive variant of Theorem \ref{Theorem:GromovTuran}.
	
	\begin{theorem}\label{Theorem:NewDensity} The following statement holds for any fixed integers $d>0$, $k>1$, $\Delta\geq 0$, $s\geq 1$, and any fixed $\delta>0$.
	
		Let $P_1,\ldots,P_k$ be $n$-point sets in $\reals^d$, and $\psi$ be a $k$-wise relation in $\Psi_{d,k,\Delta,s}$ which is satisfied for at least $\eps |P_1|\cdot\ldots\cdot |P_k|$ among the $k$-tuples $(p_1,\ldots,p_k)\in P_1\times\ldots\times P_k$. Then one can find in expected time $O\left(\sum_{i=1}^k\left(|P_i|+1/\eps\right)\log (1/\eps)\right)$, subsets $Q_i\subseteq P_i$ of cardinality $\Omega\left(\eps^{d+1+\delta}|P_i|\right)$, for $1\leq i\leq k$, so that $\psi(p_1,\ldots,p_k)=1$ holds for all $(p_1,\ldots,p_k)\in Q_1\times\ldots\times Q_k$.
	\end{theorem}

\subsection{Proof overview and further implications}
\noindent{\bf Almost $(\eps,\Psi)$-regularity.} To establish Theorems \ref{Theorem:NewRegularity} and \ref{Theorem:NewRegularitySharp}, we first obtain a preliminary partition $P=U_1\uplus\cdots\uplus U_K$, which satisfies a slightly weaker form of regularity.

\medskip
\noindent{\bf Definition.} Let $0<\eps\leq 1$, $\psi:\reals^{d\times k}\rightarrow \{0,1\}$ be a $k$-wise relation, and $P$ be a set of $n$ points in $\reals^d$. We say that a partition $P=U_1\uplus U_2\uplus \ldots\uplus U_{K}$ of $P$ is {\it almost $(\eps,\psi)$-regular} if $|U_i|\leq \eps |P|$ holds for all $1\leq i\leq K$ and, in addition, we have that 
\begin{equation}\label{Eq:Almost}
\sum  |U_{i_1}|\cdot\ldots\cdot |U_{i_k}|\leq \eps\cdot k! {n\choose k},
\end{equation}

\noindent where the sum is taken over all the $k$-families $(U_{i_1},\ldots,U_{i_k})$ that are {\it not} $\psi$-homogeneous. 

Let $\Psi$ be a family of $k$-wise relations over $\reals^d$, and $\eps>0$, and let $P$ be an $n$-point set in $\reals^d$.
We say that a partition $P=U_1\uplus\ldots\uplus U_{K}$ of is {\it almost $(\eps,\Psi)$-regular} if it is almost $(\eps,\psi)$-regular for {\it any} $k$-wise relation $\psi\in \Psi$.

\medskip
As was implicitly observed by Fox, Pach and Suk \cite[Section 4]{Regularity}, every almost $(\eps,\psi)$-regular partition can be converted, in additional $O(n)$ time, into an $\left(O(\eps),\psi\right)$-regular one; the full details of this reduction are spelled out in Section \ref{App:AlmostRegular} (Lemma \ref{Lemma:AlmostRegular}).
Thus, Theorems \ref{Theorem:NewRegularity} and \ref{Theorem:NewRegularitySharp} will be follow as immediate corollaries of the following statements. %which we establish in, respectively, Sections \ref{Sec:Main} and \ref{Sec:Generic}.

\begin{theorem}\label{Theorem:Weakly}
For any fixed integers $d>0$, $k>1$, $\Delta\geq 0$, $s\geq 1$, and any $\delta>0$, there is a constant $c=c(d,k,\Delta,s,\delta)\geq 1$ with the following property.

	Let $\eps>0$, and $P\subseteq \reals^d$ a finite point set so that $|P|\geq 1/\eps$. Then there exists an almost $(\eps,\Psi_{d,k,\Delta,s})$-regular partition $P=U_1\uplus U_2\uplus\ldots\uplus U_K$ of cardinality $K\leq c/\eps^{d+\delta}$.  
\end{theorem}

\begin{theorem}\label{Theorem:WeaklySharp}
For any fixed integers $d>0$, $k>1$, $\Delta\geq 0$, and $s\geq 1$, there is a constant $c=c(d,k,\Delta,s)>1$ with the following property.

	Let $\eps>0$, and $P\subseteq \reals^d$ a finite point set. Then one can obtain a subset $U_0$ of cardinality $|U_0|\leq c/\eps^d$, and a partition $P\setminus U_0=U_1\uplus U_2\uplus\ldots\uplus U_K$ of cardinality $K\leq c/\eps^{d}$ that is $(\eps,\psi)$-regular with respect to all the sharp $k$-wise relations $\psi\in \Psi_{d,k,\Delta,s}$ over $P$.
	
\end{theorem}

Accordingly, Theorem \ref{Theorem:Construct} will follow from the following constructive variant of Theorem \ref{Theorem:Weakly}.

\begin{theorem}\label{Theorem:WeaklyConstruct}
The partition $P=P_1\uplus P_2\uplus\ldots\uplus P_K$ of Theorem \ref{Theorem:Weakly} can be constructed in expected time $O\left(n\log (1/\eps)\right)$.
Furthermore, the construction algorithm yields a data structure that, given a relation $\psi\in \Psi_{d,k,\Delta,s}$, returns in additional $O\left(1/\eps^{dk-1+\delta}\right)$ time a collection $\G$ of $k$-families which satisfies
	\begin{equation}\label{Eq:Almost}
\sum_{(U_{i_1},\ldots,U_{i_k})\in \G}  |U_{i_1}|\cdot\ldots\cdot |U_{i_k}|\leq \eps\cdot k! {n\choose k},
\end{equation}

\noindent and includes all such $k$-families that are not $\psi$-homogeneous.\footnote{Every $k$-family $(U_{i_1},\ldots,U_{i_k})$ in the output is represented by the respective sequence $(i_1,\ldots,i_k)$.}	

\end{theorem}

%Indeed, fix  $d>0$, $k>1$, $\Delta\geq 0$, $s\geq 1$, $\delta>0$, and the point set $V$. 

\medskip
\noindent {\bf Constructing an almost-regular partition.} Let us now spell out the key ideas the underly the proofs of Theorems \ref{Theorem:NewRegularity} and \ref{Theorem:NewRegularitySharp}. The previous (almost-)regular subdivisions of semi-algebraic hypergraphs \cite{Overlap,Regularity} began with a constant-size partition and proceeded through a sequence of refinement steps.
In each round, some ``$\psi$-irregular" $k$-tuple of sets $(U_1,\ldots,U_k)$ was split so as to further reduce the ``error term" on the left hand side of (\ref{Eq:Almost}). To this end, a variant of Theorem \ref{Theorem:GromovTuran} was invoked for a $k$-partite sub-hypergraph $(U_1,\ldots,U_k,E')$ of either $(P,E)$ or its complement $\left(P,\overline{E}\right)$, depending on the ratio $\frac{|(U_1\times\ldots \times U_k)\cap E|}{|U_1|\cdot \ldots\cdot |U_k|}$.

In contrast, our partition is constructed {\it apriori} for the ground point set $P$, and is {\it oblivious to the edge set $E$}, provided that the latter is determined by a $k$-wise semi-algebraic relation $\psi\in \Psi_{d,k,\Delta,s}$. 
The more elementary partition of Theorem \ref{Theorem:NewRegularitySharp} is based on the $r$-partitioning polynomial of Guth and Katz \cite{GuthKatz} -- a $d$-variate polynomial $g$ of degree $O\left(r^{1/d}\right)$ whose zero set $Z(g)=\{x\in \reals^d\mid g(x)=0\}$ subdivides $\reals^d$ into $O(r)$ connected cells $\omega\subset \reals^d\setminus Z(g)$, each cutting out a subset $P_\omega=P\cap \omega$ of cardinality $|P_\omega|\leq |P|/r$. 
As a rule, any such partition leaves out some ``leftover" subset $U_0$ of points that fall in the zero surface $Z(g)$, and cannot be assigned to any of the ordinary subsets $P_\omega$.
To guarantee that the number of such points does not exceed $O(r)$, a generic perturbation must be applied to the ground set $P$.

Passing to $\reals^{d\times k}$ -- the ambient space of the Cartesian product $P^k$ -- yields a $(dk)$-variate polynomial $\tilde{g}:\reals^{d\times k}\rightarrow \reals$, of degree $O(kr^{1/d})=O(r^{1/d})$, with the following property: every connected cell $\tilde{\omega}$ of $\reals^{d\times k}\setminus Z(\tilde{g})$ is a Cartesian product $\omega_1\times\ldots\times \omega_k$ of some $k$ cells $\omega_i\subset\reals^d\setminus Z(g)$, and corresponds to some $k$-family $(U_1,\ldots,U_k)$ of subsets $U_i=U_{\omega_i}$.\footnote{This ``lifting" argument in $\reals^{dk}$ is admittedly inspired by a recent proof by Bukh and Vasileuski \cite{BukhVasil} of an improved estimate for the so called Same Type Lemma \cite{BaranyValtr}, which is briefly discussed in Section \ref{Sec:Discuss}. However, their ad-hoc analysis is restricted to one particular $(d+1)$-wise semi-algebraic relation, which describes orientations of $(d+1)$-point sequences, and does not directly apply to more general semi-algebraic hypergraphs.} 
If this family is {\it not} $\psi$-homogeneous for a given $k$-relation $\psi\in \Psi_{d,k,\Delta,s}$, then the respective cell $\omega$ in $\reals^{d\times k}\setminus Z(\tilde{g})$ must cross the boundary of the region $X_\psi=\{(x_1,\ldots,x_k)\in \reals^{d\times k}\mid \psi(x_1,\ldots,x_k)=1\}$.
Using standard tools from real algebraic geometry \cite{BaroneBasu,BasuSurv,BasuBook} which we lay out in Section \ref{Sec:Prelim} (most notably, the Barone-Basu bound \cite{BaroneBasu}), one can show that the number of the latter cells $\tilde{\omega}\subset \reals^{d\times k}\setminus Z(\tilde{g})$ is only $O\left(r^{(kd-1)/d}\right)=O(r^{k-1/d})$. Thus, choosing $r=c\eps^d$, with a suitable constant $c>0$, guarantees that all but an $\eps$-fraction of the $\Theta\left(r^{k}\right)=\Theta\left(1/\eps^{dk}\right)$ $k$-families $(U_1,\ldots,U_k)$ are $\psi$-homogeneous.

Unfortunately, the above partition argument does not directly apply in the more general setting of Theorem \ref{Theorem:Weakly} if many of the points of $P$ lie in the zero surface $Z(g)$ of the $r$-partitioning polynomial, and cannot be easily perturbed into a generic position without disrupting the relations $\psi(p_1,\ldots,p_k)$. Furthermore, the construction of $g$ may take $O\left(nr+r^3\right)$ time \cite{AgMaSa}, so using $r=\Theta(1/\eps^d)$ may lead to a vastly super-linear algorithm for small parameters $\eps>0$. 

To bypass these obstacles, we devise a sequence of progressively refined partitions $\Pi^0:=\{P\},\ldots,\Pi^I$. Every new partition $\Pi^i$ is obtained from its predecessor $\Pi^{i-1}$ by replacing each member set $U\in \Pi^{i-1}$, whose cardinality exceeds a certain constant threshold $n_0$, with its {\it $t$-refinement} $\Lambda_t(U)$ - a certain subdivision of $U$ into subsets $U'$, each of roughly $|U|/t$ points.\footnote{More precisely, for the sake of our amortized analysis every point $p$ is assigned a weight $\mu(p)=\mu_i(p)$ between $1$ and $t^{-d}$, so that the aggregate weight of each set $U'\in \Lambda_t(U)$ satisfies $\mu_i(U')\leq \mu_{i-1}(U)/t$.} Each refinement $\Lambda_t(U)$  is attained by the means of a {\it refinement tree} $\T_t(U)$ which loosely resembles the ``constant fan-out" tree structure of Agarwal, Matou\v{s}ek and Sharir for semi-algebraic range searching \cite{AgMaSa}. To this end, we choose $t$ to be a sufficiently large constant, and each tree is constructed only to a depth of $O(\log t)$. Using amortized analysis, which also takes into consideration the subsets that lie in {\it lower-dimensional surfaces}, we show that the ``irregularity term" on the left hand side of (\ref{Eq:Almost}) decreases by a factor of roughly $t^{1/d}$ with almost every new partition $\Pi^i$, and eventually drops below $\eps$.

The rest of this paper is organized as follows. %Recall that Theorems \ref{Theorem:NewRegularity} and \ref{Theorem:NewRegularitySharp} are direct corollaries of the respective ``almost-regularity" statements stated in Theorems \ref{Theorem:Weakly} and \ref{Theorem:WeaklySharp}.
In Section \ref{Sec:Prelim} we formally establish Lemma \ref{Lemma:AlmostRegular} which offers a straightforward way to convert almost $(\eps,\psi)$-regular partitions to $(O(\eps),\psi)$-regular ones. In addition, we review the essential real-algebraic machinery which underlies the proofs of both Theorems \ref{Theorem:Weakly} and \ref{Theorem:WeaklySharp}.
In Section \ref{Sec:Generic} we use the polynomial method of Guth and Katz \cite{GuthKatz} to obtain the ``generic" regularity results -- Theorems \ref{Theorem:NewRegularitySharp} and  \ref{Theorem:WeaklySharp}. The more intricate proof of Theorems \ref{Theorem:NewRegularity} and \ref{Theorem:Weakly} is postponed to Section \ref{Sec:Main}. 
The algorithmic aspects of our $(\eps,\Psi_{d,k,\Delta,s})$-regular partition are addressed in Section \ref{Section:Algorithmic}, where we establish Theorems \ref{Theorem:WeaklyConstruct}, \ref{Theorem:Construct} and  \ref{Theorem:NewDensity}.
In Section \ref{Sec:Discuss}, we briefly discuss discuss the connection between Theorem \ref{Theorem:Weakly} and the so called Same Type Lemma of B\'ar\'any and Valtr \cite{BaranyValtr}, and use our machinery to obtain an interesting ``semi-algebraic" generalization of the latter.\footnote{After an earlier version of this article had been made public, Tidor and Yu \cite{TidorYu} reported a slightly improved variant of Theorem \ref{Theorem:Weakly}, which yields an $(\eps,\Psi_{d,k,\Delta,s})$-regular partition of cardinality $O(1/\eps^d)$. Similar to our proofs of Theorems \ref{Theorem:WeaklySharp} and \ref{Theorem:Weakly}, they too combine a ``lifting" argument in $\reals^{dk}$ with the Barone-Basu bound \cite{BaroneBasu}. However, their partition is {\it not} accompanied by a near-linear construction algorithm. A brief comparison between the two techniques is too included in Section \ref{Sec:Discuss}.}

\section{Preliminaries}\label{Sec:Prelim}

\subsection{A reduction to almost-regularity}\label{App:AlmostRegular}
\begin{lemma}\label{Lemma:AlmostRegular}
Let $\eps>0$, $P\subset \reals^d$ be an $n$-point set, $\psi:\reals^{d\times k}\rightarrow \{0,1\}$ be a $k$-wise relation in $\reals^d$, and $P=U'_1\uplus\ldots\uplus U'_{K'}$ be an almost $(\eps/10,\psi)$-regular partition of $P$.
 Then one can obtain, in $O(n)$ time, an $(\eps,\psi)$-regular partition $P=U_1\uplus\ldots\uplus U_{K}$ of cardinality $K=O\left(K'/\eps\right)$. 
 
Moreover, the new partition does not depend on the relation $\psi$. Furthermore,
 all but an $\eps$-fraction of its $k$-families $(U_{i_1},\ldots,U_{i_k})$ are not only $\psi$-homogeneous but also each of their parts $U_{i_j}$ is contained in a single set $U'_{i_l}$.
\end{lemma}
\begin{proof}
	If $|P|\leq 10K'/\eps$ then the partition $P=U_1\uplus \ldots\uplus U_K$ is comprised of $|P|$ singleton sets, so that any $k$-family is $\psi$-homogeneous. Assume, then, that $|P|>10K'/\eps$, and let $m=\lfloor\eps |P|/(10K')\rfloor$. %We say that a part $U_i$ if {\it full} if its cardinality satisfies $|U_i|\geq m$.
Then each part $U_i$ is first subdivided into $\lfloor |U_i|/m\rfloor$ {\it red} subsets of size $m$ each, plus at most one ``leftover" part of size smaller than $m$. The ``leftover" parts are then merged and reassigned into {\it blue} subsets of size between $m$ and $m+1$. Altogether, this yields a partition of $P$ into a total of $K=O\left(|V|/m\right)=O(K/\eps)$ parts $U_{1},\ldots,U_{K}$ whose cardinalities vary between $m$ and $m+1$.

%Suppose that $\P$ is $(\eps,\psi)$-refined for a $k$-wise relation $\psi$ over $\reals^d$.
To see the $(\eps,\psi)$-regularity of the resulting partition $P=U_1\uplus\ldots\uplus U_K$, it suffices to show that 
\begin{equation}\label{Eq:NonHomogeneous}
	\sum |U_{i_1}|\cdot\ldots\cdot |U_{i_k}|\leq \eps\cdot k!{n\choose k},
\end{equation}

\noindent where the sum is again taken over all the $k$-families that are {\it not} $\psi$-homogeneous.
To this end, we say that a $k$-tuple $(p_1,\ldots,p_k)\in P^k$ of distinct vertices is {\it regular} if it comes from a $\psi$-homogeneous $k$-family $(U_{i_1},\ldots,U_{i_k})$; otherwise, such a $k$-tuple is called {\it special}.
Thus, the left hand side of (\ref{Eq:NonHomogeneous}) is bounded by the number of the special $k$-tuples $(p_1,\ldots,p_k)$, each of them falling into (at least) one of the following categories:

\begin{enumerate}
	\item[(i)] at least one of their vertices $p_i$ belongs to a blue part $U_{j_i}$,
	\item[(ii)] some pair $p_{j},p_{l}$ of their vertices come from the same part $U'_{i_j}=U'_{i_{l}}$ of the initial partition, or
	\item[(iii)] the ``ambient" parts $U'_{i_1},\ldots,U'_{i_k}$ of the vertices $p_1,\ldots,p_k$ comprise a $k$-family $(U'_1,\ldots,U'_k)$ that is not $\psi$-homogeneous.
\end{enumerate}

\noindent Since the blue parts encompass a total of at most $K'\cdot m=\eps\cdot n/10$ vertices of $P$, there exist at most $\eps {n\choose k}\frac{k!}{10}$ special $k$-tuples of the first type. Furthermore, since every part $U'_i$ ecompasses at most $\eps|P|/10$ vertices, there exist at most $\eps{n\choose k}\frac{k!}{10}$ special $k$-tuples of type (ii). Lastly, the number of the special $k$-tuples of type (iii) is at most $\eps {n\choose k}\frac{k!}{10}$ by the lemma hypothesis.
\end{proof}

\subsection{Semi-algebraic sets}
  
\medskip
\noindent{\bf Definition.} Let $d$ and $k$ be positive integers.

\begin{enumerate}
\item  A {\it real $d$-variate polynomial $f:\reals^{d}\rightarrow \reals$}, in real variables $x_1,\ldots,x_{d}$, is a function of the form
$$
f(x_1,\ldots,x_{d})=\sum_{i_1,\ldots,i_{d}\in {\mathbb N}}a_{i_1,\ldots,i_{d}}x^{i_1}\cdot \ldots \cdot x^{i_{d}},
$$
\medskip
with real coefficients $a_{i_1,\ldots,i_{d}}$. In the sequel, we use $\reals[x_1,\ldots,x_{d}]$ denote the space of all such real polynomials.
The {\it degree} of $f\in \reals[x_1,\ldots,x_{d}]$ is $deg(f)=\max\left\{\sum_{j=1}^{d}i_j\mid  a_{i_1,\ldots,i_{d}}\neq 0\right\}$.
%Thus, the real polynomials $f:\reals^{d}\rightarrow \reals$ with $deg(f)\leq D$ comprise a vector space of dimension ${d+D\choose d}$ -- the number of possible monomials $x_1^{i_1}\ldots x_{d}^{i_{d}}$ with $0\leq i_1+\ldots +i_{d}\leq D$.

	\item A {\it semi-algebraic description $(f_1,\ldots,f_s;\Phi)$} within $\reals^{d}$ is comprised of a finite sequence $f_1,\ldots,f_s\in \reals[x_1,\ldots,x_{d}]$ of real polynomials, and a Boolean formula $\Phi$ in $s$ variables (where $s$ is also the number of the real polynomials in the sequence). 
	 The {\it complexity} of this description is the pair $(\Delta,s)$, where $\Delta=\max\{deg(f_i)\mid 1\leq i\leq s\}$.

\item A subset $A\subseteq \reals^{d}$ {\it has semi-algebraic  description $(f_1,\ldots,f_s;\Phi)$} if we have that $A=\{x\in \reals^{d}\mid \Phi(f_1(x)\leq 0,\ldots,f_s(x)\leq 0)\}$. Thus, a $k$-partite $k$-uniform hypergraph $(V_1,\ldots,V_k,E)$ in $\reals^{d}$
 {admits the semi-algebraic description $(f_1,\ldots,f_s;\Phi)$} if and only if the set $E$ is cut out (as a subset of $V_1\times \ldots\times V_k$) by the set $Y\subseteq \reals^{d\times k}$ that meets the description $(f_1,\ldots,f_s;\Phi)$.%\footnote{As a matter of fact, this framework extends to $k$-partite hypergraphs $(V_1,\ldots,V_k,E)$ whose vertex sets $V_1,\ldots,V_k$ may overlap.}

%As was mentioned in the Introduction, this definition naturally extends to the $k$-uniform hypergraphs $(V,E)$ in $\reals^d$ that are not apriori $k$-partite, by insisting that $E=\left\{\{p_1,\ldots,p_k\}\in {V\choose k}\mid \overline{\left[v_1,\ldots,v_k\right]}\subseteq Y\right\}$.

%and semi-algebraic representations $(f_1,\ldots,f_s;\Phi)$ 
%whose induced subsets $A=\{x\in \reals^{d\times k}\mid \Phi(f_1(x)\leq 0,\ldots,f_s(x)\leq 0)\}$ are invariant to any permutation of the $k$ columns.

%\item We say that this description of the $k$-uniform hypergraph $(V_1,\ldots,V_k,E)$ (or $(V,E)$) in $\reals^d$ by $(f_1,\ldots,f_s;\Phi)$ is {\it sharp} if there is $\eta>0$ so that any $\eta$-perturbation $(V'_1,\ldots,V'_k,E')$ %(resp.,  $(V',E')$) 
%too meets the description $(f_1,\ldots,f_s;\Phi)$. (That is, the edge set $E$, treated as a point set in $\reals^{d\times k}$, is contained in the interior of the semi-algebraic set $Y$ which is determined by the description $(f_1,\ldots,f_s;\Phi)$.)

\end{enumerate}

We say that a semi-algebraic set $X$ is {\it crossed} by another set $Y$ if we have that $X\cap Y\neq \emptyset$ yet $X\not\subseteq Y$.

\subsection{Polynomial partitions} \label{Subsec:Polynomial}

\noindent{\bf Definition.} Let $f\in \reals[x_1,\ldots,x_d]$ be a polynomial. We refer to the hypersurface 
$$
Z(f):=\{x\in \reals^{d}\mid f(x)=0\}
$$ 
\noindent as the {\it zero set} of $f$.
The connected components of $\reals^d\setminus Z(f)$ are called {\it cells}. 

More generally, a finite collection $\{f_1,\ldots,f_s\}\in \reals[x_1,\ldots,x_d]$ of $d$-variate polynomials yields a family of hypersurfaces $\F=\{Z(f_1),\ldots,Z(f_s)\}$. The {\it arrangement} $\A(\F)$ of $\F$ \cite{BasuBook,SA} is a decomposition of $\reals^d$ into faces, that is, maximal connected sets $\tau$ with the property that $\sign (f_i(x))$ is invariant over all $x\in \tau$, for all $1\leq i\leq s$.   
The properties of $\A(\F)$ overly resemble and generalize those of a hyperplane arrangement. In particular, the cells of $\A(\F)$ are the contiguous $d$-dimensional components of $\reals^d\setminus \left(\bigcup_{i=1}^s Z(f_i)\right)$ which, in fact, are exactly the cells of $\reals^d\setminus Z\left(f_1\cdot f_2\cdot\ldots\cdot f_s\right)$; each of them admits a semi-algebraic description whose complexity is bounded in the terms of $s$ and $D=\max\{ \deg(f_i)\mid 1\leq i\leq s\}$, as described in Theorem \ref{Theorem:ComplexityCell} below.

Let $P$ be a finite point set, and $r>0$ an integer. We say that $f\in \reals[x_1,\ldots,x_d]$ is an {\it $r$-partitioning polynomial for $P$} if any cell of $\reals^d\setminus Z(f)$ encompasses at most $n/r$ points of $P$.

\begin{theorem}[The Polynomial Partition Theorem \cite{GuthKatz}]\label{Thm:PolynomialPartition}
	Let $P$ be a finite point set in $\reals^{d}$, and $1\leq r$ an integer. Then there is an $r$-partitioning polynomial $g\in \reals[x_1,\ldots,x_{d}]$ with $deg(g)=O\left(r^{1/d}\right)$. Furthermore, the zero set $Z(g)$ subdivides $\reals^d\setminus Z(g)$ into at most $c_dr$ connected regions, where $c_d$ is a constant that depends on the dimension $d$. Furthermore, such a polynomial can be computed in expected time $O(nr+r^3)$ \cite{AgMaSa}.
	\end{theorem}

The second part of Theorem \ref{Thm:PolynomialPartition} is established using a variant of the Milnor-Thom Theorem, due to Warren \cite{Warren}.
The broad usefulness of $r$-partitioning polynomials for geometric divide-and-conquer stems from the following general property established by Barone and Basu \cite{BaroneBasu}. 

\begin{theorem}\label{Theorem:ZonePolynomial}
Let $\V$ be a $l$-dimensional algebraic variety in $\reals^{d}$ that is defined by a finite set of polynomials $\G\subset \reals[x_1,\ldots,x_{d}]$, each of degree at most $\Delta$, and let $\F$ be a set of $s$ polynomials of degree at most $D\geq \Delta$. Then there exist $O\left(\Delta^{d-l}(sD)^l\right)$ faces of all dimensions in $\A(\F\cup \G)$ that are contained in $\V$; the implicit constant of proportionality may depend on $d$.
\end{theorem}

Applying Theorem \ref{Theorem:ZonePolynomial} to singleton collections $\F=\{f\}$ and $\G=\{g\}$, with $f,g\in \reals[x_1,\ldots,x_d]$, and noticing that each cell of $Z(\{f\})$ that crosses $Z(\{g\})$ corresponds to a face of $\A(\{f,g\})$ that is contained in $Z(g)$, yields the following property.

\begin{corollary}[See Lemma 4.3 in \cite{AgMaSa}]\label{Corol:CrossFewCellsPolynomial}
 Let $f,g\in \reals[x_1,\ldots,x_{d}]$ so that $deg(f)=D$ and $deg(g)=\Delta\leq D$. Then $Z(g)$ crosses $O(\Delta D^{d-1})$ open cells of $Z(f)$.	
\end{corollary}

Hence, if $f$ is the partition polynomial of degree $D=O\left(r^{1/d}\right)$ in Theorem \ref{Thm:PolynomialPartition}, then any hyper-surface whose degree $\Delta$ that is much smaller than the decomposition parameter $r$, meets roughly $O\left(r^{(1/d)(d-1)}\right)=O\left(r^{1-1/d}\right)$ open cells of $\reals^{d}\setminus Z(f)$.

\medskip
\noindent{\bf Controlling the zero-set.} To this end, for any $d\geq 1$ and $D\geq 0$, say that an ordered $t$-tuple of points $(p_1,\ldots,p_t)\in \reals^{d\times t}$ is {\it $D$-exceptional} if its elements lie in the zero set $Z(f)$ of some polynomial $f\in \reals[x_1,\ldots,x_d]$ of degree $D$. 

\begin{lemma}[\cite{SimpleProofs}]\label{Lemma:Perturbed}
Let $d\geq 1$ and $D\geq 0$ be integers, and $t:={D+d\choose d}$. Then there is a non-zero polynomial $\xi_{d,D}:\reals^{d\times t}$ with coefficients in the variables $z_{ij}$, for $1\leq i\leq d$ and $1\leq j\leq t$, whose zero set $Z(\xi_{d,D})\subset \reals^{d\times t}$ contains all the $D$-exceptional $t$-tuples $(p_1,\ldots,p_t)$, with $p_i\in \reals^d$.\end{lemma}

\begin{corollary}\label{Corol:Perturbed}
For any dimension $d>0$ there is a constant $c_0$ with the following property. Let $\eta>0$ and $P$ be an $\vartheta$-perturbed point set, where an independent $\vartheta$-perturbation is applied independently to each point of $P$.
Then the zero set $Z(f)$ of an $r$-partitioning polynomial $f$ that is constructed for the perturbed set $P$, contains at most $c_0r$ of the perturbed points of $P$.
\end{corollary}
\begin{proof}
	Since the measure of the surface $Z(\xi_{d,D})\subseteq \reals^{d\times t}$ in Lemma \ref{Lemma:Perturbed} is $0$ (see, e.g., Lemma \ref{Lemma:GoodDirection} in the sequel), it follows that, with probability $1$, no ${D+d\choose d}$ of the perturbed points of $P$ lie in the same hypersurface of $Z(f)$ of degree $D=O(r^{1/d})$, where $D$ denotes the degree of the $r$-partitioning polynomial.
	\end{proof}

%\medskip
%\noindent{\bf Definition.} Let $P$ be a finite point set in $\reals^{d}$, and $r>0$ an integer. We say that an $r$-partitioning polynomial $f\in \reals[x_1,\ldots,x_{d}]$ for %$P$ is {\it generic} if we have that $|Z(f)\cap P|\leq 2r$.

%\medskip
%\noindent{\bf Definition.} We say that a point set $P$ in $\reals^{d}$ is {\it separable} if for any integer $r>0$, and {\it any subset} $Q\subseteq P$, there is a generic $r$-partitioning polynomial $f\in \reals[x_1,\ldots,x_{d}]$ for $Q$ with $\deg(f)=O\left(r^{1/d}\right)$.

%\medskip
 %In Appendix \ref{App:PolynomialPerturbed}, we establish the following more generic variant of Theorem \ref{Thm:PolynomialPartition}.
 
%\begin{theorem}\label{Theorem:PolynomialPerturbed} %For any dimension $d\geq 1$ there is $\vartheta=\vartheta(d)>0$ with the following property.
	%Let $P$ be a finite point set in $\reals^{d}$. Then for any $\vartheta>0$, a random $\vartheta$-perturbation of $P$ is separable with probability $1$.
%\end{theorem}

\medskip
\noindent {\bf Computing an arrangement of surfaces.} The following general result \cite[Theorem 16.11]{BasuBook} (also see \cite[Theorem 4.1]{AgMaSa} and \cite[Proposition 2.3]{EstherSearching}) implies that each cell in the partition $\reals^{d}\setminus Z(f)$ of Theorem \ref{Thm:PolynomialPartition}
is semi-algebraic.

\begin{theorem}[\cite{BasuBook}]\label{Theorem:ComplexityCell}
 Let $\F=\{f_1,\ldots,f_s\}$ be a set of $s$ polynomials in $\reals[x_1,\ldots,x_{d}]$ of degree at most $D$ each. 
 Then the arrangement of their zero sets has $s^dD^{O(d)}$ faces, and can be computed in time $T=s^{d+1}D^{O(d^4)}$. Each face is assigned a semi-algebraic description whose complexity is bounded by $\left(D^{O(d^3)},T\right)$. Moreover, the algorithm computes a reference point in each face along with the adjacency information of the cells, indicating which cells are contained in the boundary of each cell.
 
 %Then the set $\reals^{d}\setminus\bigcup_{i=1}^s Z(f_i)$ is comprised of $s^{d}D^{O(d)}$ cells. Furthermore, each of these cells admits a semi-algebraic description whose complexity is bounded by $\left(D^{O(d^3)},s^{d+1}D^{O\left(d^4\right)}\right)$.
\end{theorem}

\subsection{Singly exponential quantifier elimination}\label{Subsec:SinglyExponential}
The proof of Theorem \ref{Theorem:NewRegularity} will use the following general property, known as {\it singly exponential quantifier elimination} \cite[Theorem 2.27]{BasuSurv} (also see \cite[Proposition 2.6]{EstherSearching}).

\begin{theorem}\label{Theorem:Elimination}
Let $s>1$, $a$ and $b$ be non-negative integers, and $G_1,\ldots,G_s\in \reals[x_1,\ldots,x_a,y_1,\ldots,y_b]$ be $s$ real polynomials in the variables $x_i$ and $y_j$, for $1\leq i\leq a$ and $1\leq j\leq b$, each of degree at most $D$. Let $\Upsilon$ be a boolean function in $s$ variables, and $\theta=(\theta_1,\ldots,\theta_s)\in \{-1,0,1\}^s$ a sign vector. Consider the formula
$$
\Phi(y)=\left(\exists x\in \reals^{a}\right) \Upsilon\left(\sign\left(G_1(x,y)\right)=\theta_1,\ldots,\sign\left(G_s(x,y)\right)=\theta_s\right)
$$
\noindent in the variable $y\in \reals^b$. Then there exists an quantifier-free formula $\Phi'(y)$ that is equivalent to $\Phi(y)$, of the form
$$
\Phi'(y)=\bigvee_{i=1}^I\bigwedge_{j=1}^{J_i}\left(\bigvee_{l=1}^{L_{i,j}}\sign\left(G'_{i,j,l}(y)\right)=\theta_{i,j,l}\right),
$$
where $G_{i,j,l}\in \reals[y_1,\ldots,y_b]$ is a polynomial of degree $\deg(G_{i,j,l})=D^{O(a)}$, $\theta_{i,j,l}\in \{-1,0,1\}$, and we have that
$I\leq s^{(a+1)(b+1)}D^{O(ab)}$, $J_i\leq s^{a+1}D^{O(a)}$, and $L_{i,j}\leq D^{O(a)}$, for all $1\leq i\leq I$, $1\leq j\leq J_i$, and $1\leq l\leq L_{i,j}$.

Moreover, there is an algorithm for computing $\Phi(y)$ in time $s^{(a+1)(b+1)}D^{O(ab)}$.
\end{theorem}

\noindent{\bf Remark.} Any condition of the form $\sign\left(G_{i,j,k}(y)\right)=\theta_{i,j,l}$ is equivalent to a Boolean combination of the conditions $(G_{i,j,k}(y)\leq 0)$ and $(-G_{i,j,k}(y)\leq 0)$. 
Hence, the set $\{y\in \reals^d\mid \Phi'(y)\}$ that is defined by the quantifier free formula $\Phi'(y)$ in Theorem \ref{Theorem:Elimination}, admits a semi-algebraic description whose complexity is bounded by $\left(D^{O(a)},s^{(a+1)(b+1)+a+1}D^{O(ab)}\right)$.

\subsection{Decomposing a hypersurface into monotone patches}

\noindent{\bf Definition.} Let $f\in \reals[x_1,\ldots,x_d]$. We say that that a direction $\vec{\theta}\in \SS^{d-1}$ is {\it good} for $f$ if any $\vec{\theta}$-oriented line in $\reals^d$ intersects the zero set $Z(f)$ at finitely many points.

\begin{lemma}\label{Lemma:GoodDirection}[\cite{SchwartzSharir}, pp. 304--305 and pp. 314--315]
Let $f\in \reals[x_1,\ldots,x_d]$. Then a uniformly chosen direction $\vec{\theta}\in \SS^{d-1}$ is good for $f$ with probability $1$.	
\end{lemma}

The following two properties were established by Agarwal, Matou\v{s}ek and Sharir \cite[Theorems 6.2 and 6.3]{AgMaSa}.

\begin{lemma}\label{Lemma:Patches}
	Let $f\in \reals[x_1,\ldots,x_d]$ be a $d$-variate polynomial of degree $D$, and suppose with no loss of generality that the direction of the $x_d$-axis is good for $f$. Then $Z(f)$ can be decomposed, in $D^{O\left(d^4\right)}$ time, into $D^{O(d)}$ $x_d$-monotone semi-algebraic patches, each of description complexity at most $\left(D^{O\left(d^3\right)},D^{O\left(d^4\right)}\right)$.  
\end{lemma}

\section{Proof of Theorems \ref{Theorem:NewRegularitySharp} and \ref{Theorem:WeaklySharp}}\label{Sec:Generic}

%Hence, the rest of this section is dedicated to the proof of Theorem \ref{Theorem:WeaklySharp}.

\medskip
\noindent{\bf Proof of Theorem \ref{Theorem:WeaklySharp}.} Let us fix the quantitites $d>0$, $k>1$, $\Delta\geq 0$, $s\geq 1$, and the $n$-point set $P$. 
Since the semi-algebraic relations $\psi\in \Psi_{d,k,\Delta,s}$ can determine only finitely many distinct subsets 
$$
\P_\psi=\{(p_1,\ldots,p_k)\in {P^{\times k}}\mid \psi(p_1,\ldots,p_k)=1\},
$$
where $P^{\times k}$ denotes the collection of all the $k$-sequences $(p_1,\ldots,p_k)$ of distinct elements $p_i\in P$, choosing a sufficiently small $\vartheta=\vartheta(P)>0$ guarantees that no such subset $\P_\psi$, that is determined by a sharp $k$-wise relation $\psi\in\Psi_{d,k,\Delta,s}$ over $P$, is ``altered" by an $\vartheta$-perturbation of $P$. 

%(This is possible because the overall number of such subsets $Y_\chi\cap V^k$ is only finite, and each of them can be assigned a unique sharp $k$-wise relation $\chi\in \overline{\Psi}_{d,k,\Delta,s}$.)

Fix $r=\lceil\tilde{c}/\eps^{d}\rceil$, where $\tilde{c}=\tilde{c}(d,k,\Delta,s)>1$ is a sufficiently large constant to be fine-tuned in the sequel. Applying Theorem \ref{Thm:PolynomialPartition} to the $\vartheta$-perturbed set $P$ yields an $r$-partitioning polynomial $g\in \reals[x_1,\ldots,x_d]$.  
Recall that the complement $\reals^d\setminus Z(g)$ is comprised of $K\leq c_dr$ pairwise-disjoint cells $\omega_{1},\ldots,\omega_{K}$, so that each cell $\omega_{i}$ cuts out a subset $U_{i}=P\cap \omega_{i}$ of size at most $|P|/r\leq \eps^d n/\tilde{c}^d\leq \eps|V|$. 
Furthermore, Corollary \ref{Corol:Perturbed} implies that the points of $P\cap Z(g)$ comprise the ``zero-part" $U_{0}$ whose cardinality $n_0$ satisfies $n_0\leq c_0r\leq 2c_0\tilde{c}/\eps^d$, where $c_0$ denotes a suitable constant that depends only on the dimension $d$.
It can be assumed, with no loss of generality, that $n\geq 4c_0\tilde{c}/\eps^d\geq 2n_0$, or else our almost-regular partition is comprised of a single set $U_0=P$ (to which end we will choose $c\geq 4c_0\tilde{c}$).

%Notice that there exist only $O\left(r^{k-1}\right)$ $k$-tuples $\left(U_{1,j_1},\ldots,U_{k,j_k}\right)$
Fix a $k$-wise relation $\psi\in \Psi_{d,k,\Delta,s}$ whose semi-algebraic description $(f_1,\ldots,f_s,\Phi)$ satisfies $\deg(f_i)\leq \Delta$ for all $1\leq i\leq s$. 
It suffices to show that, given a suitably large choice of the constant $\tilde{c}>0$, we have that
\begin{equation}\label{Eq:Generic}
	\sum |U_{i_1}|\cdot\ldots\cdot |U_{i_k}|\leq \eps k!{n-n_0\choose k}
\end{equation}
where the sum on the left hand side is taken over all the $k$-families $(U_{i_1},\ldots,U_{i_k})$, with distinct indices $1\leq i_j\leq K$, that are not $\psi$-homogeneous.

To this end, we ``lift" $g$ so as to subdivide the space $\reals^{d\times k}$ of the $k$-tuples $(\hat{x}_1,\ldots,\hat{x}_k)$ of point $\hat{x}_i\in \reals^d$. 
For each $1\leq i\leq k$ we define the polynomial $\tilde{g}_i\in \reals[x_{1,i},\ldots,x_{d,i}]$ by applying $g$ to the coordinates of the $i$-th point $\hat{x}_i$; that is $\tilde{g}_i:=g\left(x_{1,i},\ldots,x_{d,i}\right)$. Then the product 
$$
\tilde{g}:=\prod_{i=1}^k \tilde{g_i}=\prod_{i=1}^k g\left(x_{1,i},\ldots,x_{d,i}\right),
$$ 
\noindent  is a polynomial in the $dk$ coordinates of $\reals^{d\times k}$, whose degree satisfies $\deg\left(\tilde{g}\right)=O\left(kr^{1/d}\right)=O\left(r^{1/d}\right)$.

Let $\left(U_{i_1},\ldots,U_{i_k}\right)$ be a $k$-family that is {\it not} $\psi$-homogeneous, and that does not include $U_0$. 
Fix a pair of $k$-tuples $\hat{p}=(p_1,\ldots,p_k)\in U_{i_1}\times\ldots\times U_{i_k}$ and $\hat{q}=(q_1,\ldots,q_k)\in U_{i_1}\times\ldots\times U_{i_k}$, with $\psi(p_1,\ldots,p_k)\neq \psi(q_1,\ldots,q_k)$.
Since $\omega_{i_1},\ldots,\omega_{i_k}$ are path connected cells in $\reals^d\setminus Z(g)$, so is their Cartesian product 
$$
\tilde{\omega}_{i_1,\ldots,i_k}=\omega_{i_1}\times \ldots\times\omega_{i_k},
$$ 
\noindent which is easily seen to constitute a single open cell in $\reals^{d\times k}\setminus Z\left(\tilde{g}\right)$.  Hence, $\hat{p}$ and $\hat{q}$ must be connected by a path $\xi:[0,1]\rightarrow \tilde{\omega}_{i_1,\ldots,i_k}\subseteq \reals^{d\times k}$, which does not leave $\tilde{\omega}_{i_1,\ldots,i_k}$. 
Since $\psi(\hat{p})\neq \psi(\hat{q})$, there must be an index $1\leq h\leq s$ so that $\sign\left(f_h(p_1,\ldots,p_k)\right)\neq \sign\left(f_h(q_1,\ldots,q_k)\right)$. Hence, there must exist a $k$-tuple $(u_1,\ldots,u_k)\in {\sf Im}(\xi)$ that belongs to $\tilde{\omega}_{i_1,\ldots,i_k}\cap Z\left(f_h\right)$. 
%In particular, any such cell $\tilde{\omega}_{j_1,\ldots,j_k}$ must be intersected by the surface $Z(f_h)$ {\it in $\reals^{dk}$}. %, and is comprised of several faces in the arrangement of $\A\left(\left\{g_h,\tilde{f}\right\}\right)$. 
%We ``charge" the $k$-tuple $\left(U_{1,j_1},\ldots,U_{k,j_k}\right)$ to any of these faces $\tau$ (say, the one containing $(\tilde{v}_1,\ldots,\tilde{v}_k)$).

To recap, every ``$\psi$-irregular" $k$-family $\left(U_{i_1},\ldots,U_{i_k}\right)$ gives rise to a distinct cell $\omega_{i_1,\ldots,i_k}\subseteq \reals^{d\times k}\setminus Z(\tilde{g})$, which crosses at least one of the surfaces $Z(f_h)$. However, in view of Corollary \ref{Corol:CrossFewCellsPolynomial}, the number of the latter cells
is only 
$\displaystyle
O\left(s\Delta r^{\frac{dk-1}{d}}\right)=O\left(r^{k-1/d}\right)$, so that the left hand side of (\ref{Eq:Generic}) is only $O\left(r^{k-1/d}\cdot (n/r)^k\right)=O(n^k/r^{1/d})$. Hence, a sufficiently large choice of $\tilde{c}=\tilde{c}(d,k,\Delta,s)$ and $c=c(d,k,\Delta,s)\geq 4\tilde{c}c_0$ guarantees that, with $|P|\geq 2n_0=2|U_0|$, the partition of $P\setminus U_0=U_1\uplus\ldots\uplus U_K$ is indeed almost $(\eps,\psi)$-regular, and that its cardinality is at most $c/\eps^d$.
 $\Box$
 
 \medskip
\noindent{\bf Remark.} A brief examination of the proof of Theorem \ref{Theorem:WeaklySharp} implies that the implied dependence on the constant parameters $\Delta$ and $s$, within the bound on the partition size, is of the order $O((s\Delta)^d)$, with additional hidden constants that depend only on $d$ and $k$.

\medskip
\noindent{\bf Proof of Theorem \ref{Theorem:NewRegularitySharp}.} To deduce Theorem \ref{Theorem:NewRegularitySharp} from the statement of Theorem \ref{Theorem:WeaklySharp}, we distinguish between two cases. If $n=|P|=O\left(1/\eps^{d+1}\right)$ then our regular partition of $P$ is comprised of singleton sets $\{p\}$, for $p\in P$. Otherwise, if $|P|\geq c/(100\eps)^{d+1}$ say, the preliminary almost-regular partition is accomplished by applying Theorem \ref{Theorem:WeaklySharp} to $P$ with a slightly smaller parameter $\eps/100$, and noting that the set $U_0$ too has cardinality at most $\eps|P|/10$ (so that its contribution to the left hand side of (\ref{Eq:Almost}), for any $\psi\in \Psi_{d,k,\Delta,s}$, does not exceed $k! \cdot \eps {n\choose k}/10$). $\Box$

\section{Proof of Theorems \ref{Theorem:NewRegularity} and \ref{Theorem:Weakly}}\label{Sec:Main}
If $|P|\leq 1/\eps$, then the regular partition is again comprised of $n$ singletons $\{p\}$, with $p\in P$. 
Otherwise, the statement of Theorem \ref{Theorem:NewRegularity} will follow by applying the refinement of Lemma \ref{Lemma:AlmostRegular} to an almost $(\eps/10,\Psi_{d,k,\Delta,s})$-regular partition $P=U_1\uplus \ldots \uplus U_K$, which can be obtained via Theorem \ref{Theorem:Weakly}. Hence, it suffices to establish the latter theorem.

To this end, let us fix $\delta':=\delta/(100dk)$ and an additional constant $\eta=\delta'/(100dk)$.
For the sake of brevity, we adopt the notation $x\lll_\eta y$ whenever $x=O\left(y^{\eta/(10dk)}\right)$ holds for positive integers $x,y$.\footnote{If $x,y\in (0,1)$, this notation means that $x=O\left(y^{\frac{10dk}{\eta}}\right)$.}
To facilitate our construction and its subsequent analysis, we also fix a suitably large constant $t=t(d,k,\Delta,s,\eta)$.

Recall that for the sake of almost $(\eps,\psi)$-regularity, it is necessary that the cardinality of each part $U\in \Pi(P,\eps)$ does not exceed $\eps |P|$, which is possible only for the point sets $P$ of cardinality at least $1/\eps$. As a matter of fact, the $O(1/\eps^{d+\delta})$-size partition $\Pi(P,\eps)$ in the sequel will apply to all finite sets $P$, while still yielding the inequality (\ref{Eq:Almost}) for all $\psi\in \Psi_{d,k,\Delta,s}$, and the bound $|U|\leq \max\{1,\eps^{d+\delta'} |P|\}$ on the respective cardinalities of all member sets $U\in \Pi(P,\eps)$.

\medskip
\noindent{\bf The intermediate partitions $\Pi^i$.} To obtain an almost $(\eps,\Psi_{d,k,\Delta,s})$-regular partition $\Pi(P,\eps)=\Pi_{d,k,\Delta,s,\delta}(P,\eps)$ of $P$, we construct a sequence of $I+1\leq \log_t \left(1/\eps^{d+\delta'}\right)+i_0$ progressively refined partitions $\Pi^0=\{P\},\ldots,\Pi^I=\Pi(P,\eps)$, where $i_0>d$ is a suitable constant to be determined in the sequel. 

\medskip
\noindent{\it The intrinsic dimension $l(U)$.} For each $0\leq i\leq I$, every member set $U$ of the $i$-th partition $\Pi^i=\Pi^i(P)$ is assigned {\it intrinsic dimension $l(U)\in \{0,\ldots,d\}$}. 
 It will maintained that $l(U)=0$ holds if and only if $|U|=1$, and otherwise the cardinality $|U|$ is bounded from below by a certain constant $n_0=n_0(d,k,\Delta,s,\delta)$ that will be determined in the sequel. 

Every element $p\in U$ is assigned weight $\mu(p)=\mu_i(p):=t^{l(U)-d}$, which is bound to decline whenever $p$ is ``passed on" to a subset $W\in \Pi^{i+1}$ of smaller intrinsic dimension $l(W)<l(U)$. 
In addition, we will use $\mu(U)=\mu_i(U)$ to denote the aggregate weight $\sum_{p\in U}\mu_i(p)$ of a subset $U\in \Pi^i$. 

\medskip
\noindent{\it The set $S(U)$, and the projection $\pi_S:\reals^d\rightarrow \reals^l$.} In addition to $l(U)$, every set $U\in \Pi^i$ will be assigned 

\begin{enumerate}
	\item an auxiliary semi-algebraic set $S=S(U)$, whose description complexity is bounded by some constants $(\Delta_l,s_l)$ that depend on $d,k,l,s,\Delta$ and $\delta$, and 
	\item an affine projection function $\pi=\pi_S:\reals^d\rightarrow \reals^{l}$ which is injective over $S(U)$.
\end{enumerate}
 
 Specifically, if $l(U)=0$ then $U=S(U)$ is a singleton $\{p\}$, whose semi-algebraic description complexity is clearly bounded by the pair $(1,d)$ (as $U$ is an intersection of $d$ hyperplanes in $\reals^d$), and the projection $\pi_S$ trivially sends $p$ to $0\in \reals^0$.

%For $\Pi^0=\{P\}$, we set $l(P):=d$ and $S(P):=\reals^d$.
%As we proceed from $\Pi^i$ to its successor partition $\Pi^{i+1}$, every set $U\in \Pi^i$ is replaced with either one or several pairwise disjoint subsets $W\in \Pi^{i+1}$ which satisfy $S(W)\subseteq S(U)$, and whose intrinsic dimensions $l(W)$ do not exceed $l(U)$.

It can be assumed, with no loss of generality, that the cardinality of $P$ is larger than $n_0$ or, else, our almost $(\eps,\Psi_{d,k,\Delta,s})$-partition is comprised of at most $n_0=O(1)$ singletons.
The construction begins with $\Pi^0=\{P\}$, $l(P)=d$ and $S(P)=\reals^d$. %where we set $\X(V)=\Y(V)=\reals^d$, and $\pi_V\equiv id_{\reals^d}$. 
Every subsequent partition $\Pi^{i+1}$ is derived by replacing every element $U\in \Pi^{i}$ with its so called {\it $t$-refinement} $\Lambda_t(U)$ -- a subdivision of $U$ into $O(t^{1+\eta/(10k)})$ pairwise disjoint subsets $W\in \Lambda_t(U)$ whose intrinsic dimensions assume values $l(W)\in \{l,l-1,0\}$, and their aggregate weights satisfy $\mu_{i+1}(W)\leq \mu_i(U)/t$. If $l(U)>0$, we have that $S(W)\subseteq S(U)$ for all subsets $W\in \Lambda_t(U)$ of positive intrinsic dimension $l(W)>0$, and equality $S(U)=S(W)$ arises if and only if $l(U)=l(W)$. 

\medskip
In what follows, we use $\F^i_k=\F^i_k(P)$ to denote the collection of all the $k$-families $(U_1,\ldots,U_k)$ within $\Pi^i$. Given a $k$-wise relation $\psi\in \Psi_{d,k,\Delta,s}$,  the ``$\psi$-irregularity" of the $i$-th partition $\Pi^i$ of $P$ will be measured using the following amortized function

$$
\rho(\Pi^{i},\psi):=\sum \mu\left(U_1\right)\cdot \ldots\cdot \mu\left(U_k\right),
$$

\noindent where the sum is taken over all the $k$-families $(U_{1},\ldots,U_{k})\in \F^i_k$ that are not $\psi$-homogeneous; in what follows, we denote the latter collection by $\G^i_k=\G^i_k(P,\psi)$.

 % that are {\it not} homogeneous with respect to the given relation $\psi\in \Psi_{d,k,\Delta,s}$. 

%\medskip
%For each $0\leq i\leq J$, the $i$-th partition $\Pi^i$ will meet the following criteria. 

%\begin{enumerate}
%    \item $\Pi^i$ is comprised of $|\Pi^i|\leq t^{i+\eta}$ parts, each part $U$ of weight $n/t^{}\leq n/t^{i-d}$.

%	\item For every semi-algebraic relation $\psi\in \Psi_{d,k,\Delta,s}$, we have that 
	%$$
%\mu(\P^i,\psi)\leq \frac{k!}{t^{i/d-i\delta}}{n\choose k}.
	%$$
	
%\end{enumerate}

% We say that a set $U\in \P^i(V)$ is {\it ordinary} if it has positive intrinsic dimension $1\leq l\leq d$ which is either (i) equal to $d$ or (ii) same as the intrinsic dimension of its ``parent set" in $\P^{i-1}(V)$.\footnote{Thus, the set $V\in \P^0(V)$ is considered ordinary.} Any other set $U\in \P^i(V)$ will be called {\it special}.

%Furthermore, if the set $U\in \P^i(V)$ is ordinary, it will be assigned another, {\it $l$-dimensional} semi-algebraic set $\O(U)\in \Gamma_{l,a_l,a_l\log_{r_l}t}$ which contains the projection $\pi_\S(U)\subseteq \reals^l$. 
%$\max\{\Delta',s'\}\leq \left(r_{d'+1}+\Delta+s\right)^{O(d^4)}$. 

\subsection{The $t$-refinement $\Lambda_t(U)$}\label{Subsec:Refinement}

At the center of our almost $(\eps,\Psi_{d,k,\Delta,s})$-regular partition lies the notion of the {\it $t$-refinement} $\Lambda_t(U)$ of a point set $U\subset \reals^d$. 
It will be assumed in the sequel that the set $U$ is ``equipped" with a pre-assigned semi-algebraic set $S=S(U)\supseteq U$, whose description complexity is bounded by a pair of constants $(\Delta_l,s_l)$, and with a projection function $\pi=\pi_S:\reals^d\rightarrow \reals^l$ that is injective over $S$.\footnote{Besides these quantities $l(U)$, $S=S(U)$, $\Delta_l,s_l$ and $\pi_S$, the $t$-refinement $\Lambda_t(U)$ is self-contained in the sense that it does not depend on the superset $P$, or its partition $\Pi^i$ that includes $U$.}

If $U$ has intrinsic dimension $l(U)=0$, then we have that $|U|=\{p\}$ for some $p\in P$. Hence, we set $\Lambda_t(U)$ consists of the singleton $U$, which is assigned the same intrinsic dimension $0$, same set $S(U):=U$, and same projection function $\pi_S$ mapping $\reals^d$ to $0\in \reals^0$. Hence, it can be assumed in what follows that $l(U)>0$ and, therefore, $|U|>n_0$.

The refinement $\Lambda_t(U)$ will be constructed through repeated application of the polynomial $r$-partition of Theorem \ref{Thm:PolynomialPartition} in 
 $\reals^l$. 
As it will become clear in the sequel, it is instrumental to select such a partition parameter $r=r_l$ that depends on the intrinsic dimension $l=l(U)$, and sharply {\it increases} every time we proceed to refine subsets $W\in \Lambda_t(U)$ with {\it lower} intrinsic dimension $l(W)<l(U)$. To this end, we will use a sequence of $d$ auxiliary constants $r_1,\ldots,r_d$, which depend only on $d$, $k$, $\Delta$, $s$ and $\eta$, and satisfy\footnote{Most of our asymptotic bounds will be expressed in the terms of the parameter $t$, which will be chosen to be much larger than the constants $r_1,\ldots,r_d$. Thus, bounds of the form $O(t^a)$ will to used to suppress $r_l$-factors for $1\leq l\leq d$.}
$$
\max(\{c_l\mid 1\leq l\leq d\}\cup \{\Delta,s,k,d\})\lll_\eta r_d\lll_\eta r_{d-1}\lll_\eta \ldots\lll_\eta r_1,
$$

\noindent where $c_l$ denotes the constant of proportionality in the polynomial partition of Theorem \ref{Thm:PolynomialPartition} adapted to $\reals^l$. In addition,
we set $n_0:=r_1$.
Our induction in the intrinsic dimensions $l(U)$ of the sets $U$ will also involve a fictitious parameter $r_{d+1}=0$.

\medskip
\noindent{\bf The $t$-refinement tree $\T_t(U)$.} The definition of $\Lambda_t(U)$ involves an auxiliary tree-like structure $\T_t=\T_t\left(U\right)$.
% -- a partition of $U$ into at most $t^{1+\delta}$ subsets. % which ``derive" from $U$ in the $(i+1)$-st partition $\P^{i+1}(V)$.
%For the rest of this description, we will keep the triple $(X,U,\pi)\in \P^i(V)$ fixed, and use $r$ to denote the parameter $r_{d'}$.
 %, to construct a search tree $\T(U_{j}^i)$ for answering semi-algebraic queries amid $U_j^i$. 
%Namely, the subsets $W\in \R(U)$ will be chosen from among the canonical subsets $V_\zeta\subset U$ that reside at the leaf nodes $\zeta$ of $\T(U)$.
%To guarantee that the weight $w(W)$ of every subset $W$ does not exceed $w(U)/\lambda$, it will suffice to construct $\T(U)$ in a top-down fashion to a maximum depth of $O(\log \lambda)$.
Every node $\alpha$ in $\T_t(U)$ is assigned a {\it canonical subset} $W_\alpha\subseteq U$, which is set to $U$ for the root node $\alpha=\alpha_0$. 
 The rest of the construction of $\T_t(U)$ proceeds from $\alpha_0$ in top-down fashion, and bottoms out at certain {\it terminal nodes} whose canonical subsets $W_\alpha$ yield the  $t$-refinement $\Lambda_t(U)$ of $U$. 
 
 The refinement tree $\T_t(U)$ is overly inspired by the ``bounded fan-out" hierarchical decomposition of Agarwal, Matou\v{s}ek and Sharir \cite[Section 5]{AgMaSa}, with the following key differences: (a) the construction of $\T_t(U)$ is carried out only to the maximum level 
$h_l:=\log_{r_l}t$,\footnote{For the sake of brevity, we routinely use $\log_a b$ to denote the quantity $\lceil\log_a b\rceil$.} (b) the lower-dimensional subsets $W\subset U$ that fall in the zero-sets of the $r_l$-partitioning polynomials, are immediately passed on to $\Lambda_t(U)$.

To further split the canonical subset $W_\alpha$, that is assigned to an {\it inner} (i.e., non-terminal) node $\alpha$, we apply the polynomial partition of Theorem \ref{Thm:PolynomialPartition} to the $l$-dimensional projection $W^*_\alpha=\pi(W_\alpha)$, which yields an {\it $l$-variate} $r_l$-partitioning polynomial $g_\alpha\in \reals[y_1,\ldots,y_{l}]$ of maximum degree $D_l=O\left(r_l^{1/l}\right)$. (Here $y_1,\ldots,y_l$ are the $l$ coordinates of $\reals^l\supseteq \pi(S)$.) Notice that the zero set $Z_\alpha:=Z(g_\alpha)$ of $g_\alpha$ subdivides the points of $W^*_\alpha\setminus Z_\alpha$ into at most $c_lr_l$ subsets, each of cardinality is at most $|W_\alpha|/r_l$. 
(With a nominal abuse of notation, we will also use $g_\alpha$ and $Z_\alpha$ to denote, respectively, the pull-back polynomial $g_\alpha\circ \pi$ in $\reals[x_1,\ldots,x_d]$, and its zero set within $\reals^d$.)

%If the node $\nu$ is terminal, then we set $g_\nu\equiv 1_{\reals^l}$, so that $Z_\nu=\emptyset$.
%  \item $\O_\nu$ is a semi-algebraic set in $\reals^l$ which contains $\overline{U}_\nu$. %If the node $\nu$ is non-terminal, then we have that $\pi_\X(U_\nu)=\pi_\X(U)\cap \Y_\nu$.

\medskip
Every terminal node $\alpha$ 
falls into the following (non-exclusive) categories:
\begin{itemize}
\item[(i)] the level of $\alpha$ in $\T_t$ is $h_l$,
\item[(ii)] we have that $|W_\alpha|\leq n_0$, or
\item[(iii)] the $l$-dimensional projection $W^*_\alpha=\pi(W_\alpha)$ lies in the zero set $Z_{\alpha'}=Z(g_{\alpha'})$ of the parent node $\alpha'$ of $\alpha$.
\end{itemize}

%A terminal node is {\it ordinary} if it exclusively falls into the first category; any other terminal node will be called {\it special}.
%The terminal nodes of types (i) and (ii) will be called {\it ordinary}, whereas the terminal nodes of the last type (iii) will be called {\it special}.

\medskip
To construct the subtree rooted at an inner node $\alpha$, we distinguish between two classes of points: (a) the points of $W_\alpha\setminus Z_\alpha$, whose projections lie in $r'\leq c_lr_l$ connected regions (i.e., cells) $C_1,\ldots,C_{r'}\subset \reals^{l}\setminus Z_\alpha$, and (b) the points of $W_\alpha\cap Z_\alpha$.
If the interior of a region $C_j\subset \reals^{l}\setminus Z_\alpha$ contains at least one projected point of $W^*_\alpha$, then $\alpha$ is ``attached" a child node $\alpha_j$ whose canonical subset $W_{\alpha_j}:=W_{\alpha}\cap \pi^{-1}(C_j)$ satisfies $|W_{\alpha_j}|\leq |W_\alpha|/r_l$.  %In this case, $\beta$ becomes an {\it ordinary node} of $\T_t(U)$, and is assigned the respective region $\o_\beta=\o$ in $\reals^l\setminus Z_\alpha$.

%If $|U_\mu|\leq c_0$, or the level of $\mu$ is at least $\log_{r_l}t$, then $\mu$ is a terminal node of $\T_t(U)$.
%Otherwise, $\mu$ too is an inner node of $\T_t(U)$. %, in which case we set $\O_\mu:=\O_\nu\cap \omega$.
%If $\xi$ is not terminal, we proceed to construct a subtree over $V_\xi$.
%Let $\tilde{g}_\nu$ denote the ``pull-back" $d$-variate polynomial $\tilde{g}_\zeta=g_\zeta\circ \pi$ in $\reals[x_1,\ldots,x_d]$. 
%As a result, every point $v$ of $U_\nu\setminus \pi^{-1}(Z_\nu)$ has been relegated to a subtree, which we construct over the unique child subset $U_\xi$ that contains $v$.

Note that if $l(U)=1$, then the $r_l$-partitioning polynomials correspond to subdivisions of $\reals^1$ into at most $c_1r_1$ continuous intervals $\lambda$ with $|W^*_\alpha\cap \lambda|\leq |W^*_\alpha|/r_l$, and it can be assumed that $Z_\alpha\cap W_\alpha=\emptyset$.

The ``leftover" points of $W_\alpha\cap Z_\alpha$ are dispatched by 1. first applying Lemma \ref{Lemma:GoodDirection} to obtain a so called {\it good direction} ${\theta}_\alpha\in \SS^{l-1}$ with respect to the $(l-1)$-dimensional hypersurface $Z_\alpha$ {\it within $\reals^l$}, so that $|Z_\alpha\cap \ell|<\infty$ would hold for any $\theta_\alpha$-parallel line $\ell$, and then 2. invoking Lemma \ref{Lemma:Patches} to further subdivide $Z_\alpha$ into $z=r_l^{O(1)}$ $\theta_\alpha$-monotone semi-algebraic patches $Z_1,\ldots,Z_z$, whose semi-algebraic description complexities are bounded by $\left(\tilde{\Delta}_{l},\tilde{s}_{l}\right)$, with $\tilde{\Delta}_l=r_l^{O(l^2)}$ and $\tilde{s}_l=r_l^{O(l^3)}$.
As a result, every point of $W_\alpha\cap Z_\alpha$ belongs to exactly one of the semi-algebraic subsets $S_j:=S(U)\cap \pi^{-1}(Z_j)$. %, whose dimension does not exceed $l-1$.
Any such semi-algebraic set $S_j\subseteq \reals^d$, that encompasses at least one point of $W_\alpha$, gives rise to a child $\alpha'_j$ of $\alpha$ with $W_{\alpha'_j}:=W_\alpha\cap S_j$.

It is immediate to check that every point of $U$ ends up in the canonical subset $W=W_\alpha$ of exactly one terminal node $\alpha$. In what follows, we distinguish between two classes of such canonical subsets $W_\alpha$.

\medskip
\noindent{\bf Ordinary sets.} If the canonical set $W=W_\alpha$ of a terminal node $\alpha$ satisfies $|W|>n_0$, and it lies outside the zero set $Z_{\alpha'}$ of its parent node $\alpha'$, 
	 then $W_\alpha$ becomes an {\it ordinary set} of $\Lambda_t(U)$, and is assigned the same intrinsic dimension $l>0$, the same set $S(W):=S(U)$, and the same projection function $\pi_S:\reals^d\rightarrow \reals^l$.
	
	 Since such a terminal node $\alpha$ must lie at level $h=\log_{r_l}t$,  	
tracing the descending sequence $\alpha_0,\ldots,\alpha_{h-1}$ of its $h$ ancestor nodes yields $|W|\leq |U|/r_l^h=|U|/t$, so that $\mu(W)\leq \mu(U)/t$. 
	%We then have $\O(U_\nu):=\bigcap_{i=1}^{h-1}\o_i$, where $\o_i$ denotes the unique cell in $\reals^{l}\setminus Z_{\nu_i}$ that contains the canonical subset $U_{\nu_{i+1}}$. 
	
 In addition to the set $S(W)=S(U)$, every ordinary set $W=W_\alpha$ in $\Lambda_t(U)$ will be assigned a finer semi-algebraic set $A(W):=S(W)\cap \pi^{-1}_S\left(\bigcap_{j=0}^{h-1}C_j\right)$, where $C_j$ denotes the unique cell in $\reals^l\setminus Z_{\alpha_j}$ that contains $W^*$ for $0\leq j\leq h-1$.

If $l(U)=0$, then $U$ is also an ordinary set of $\Lambda_t(U)=\{U\}$, and we set $A(U)=U=\{p\}$ for some $p\in P$.

%Note the semi-algebraic description complexity of $\Q(W)$ is bounded by $(r_{})$.

\medskip	
\noindent{\bf Special sets.} If the set $W=W_\alpha$ is not ordinary, it becomes a {\it special set} of $\Lambda_t(U)$, whose intrinsic dimension is strictly smaller than $l$. As a result, the updated weight of each point $p\in W$ is at most $\mu(p)\leq t^{(l-1)-d}=\mu(p)/t$, which again yields $\mu(W)\leq \mu(U)/t$.

Specifically, if its cardinality $|W_\alpha|$ is at most $n_0$, then $W_\alpha$ ``contributes" to $\Lambda_t(U)$ $|W_{\alpha}|$ singleton subsets $\{p\}$ of intrinsic dimension $l(\{p\})=0$, one for each point $p\in W_\alpha$.

Otherwise, if $|W_\alpha|>n_0$, then $W_\alpha$ is contained in the zero set $Z_{\alpha'}=Z(g_{\alpha'})$ of its parent node $\alpha'$. Hence, its intrinsic dimension is reduced to $l-1$, and $S(W_\alpha)$ is set to $S_j:=S(U)\cap \pi^{-1}(Z_j)$, where $Z_j$ denotes the unique $\theta_{\alpha'}$-monotone semi-algebraic patch of $Z_{\alpha'}$ that contains $W^*_\alpha$. Accordingly, the new projection mapping $\pi_{S_j}$ is given by the ``$\pi_{\theta_{\alpha'}}$-pull-back" $\pi_{{\theta_{\alpha'}}}\circ\pi$ of $\pi$, which is clearly injective over $S_j$.

\begin{lemma}\label{Lemma:Refinement}
	Let $U$ be a set of positive intrinsic dimension $l=l(U)>0$, and suppose that the description complexity of $S(U)$ is at most $(\Delta_l,s_l)$. Then the $t$-refinement $\Lambda_t(U)$ of $U$ consists of $O\left(t^{1+\eta/(10k)}\right)$ sets $W$ which satisfy $\mu(W)\leq \mu(U)/t$.
	Furthermore, the following holds for each part $W\in \Lambda_t(U)$. %with positive intrinsic dimension $l(W)$.
	\begin{enumerate}
		\item If $W$ is an ordinary set, then we have that $l(W)=l(U)$ and $S(W)=S(U)$ (so that the description complexity of $S(W)$ is too bounded by $(\Delta_l,s_l)$), while the description complexity of the subset $A(W)\subseteq S(W)$ is bounded by $\left(\max\{\Delta_l,r_l^c\},s_l+r_l^c\log t\right)$.
		\item If $W$ is a special set, then we have that $l(W)=l-1$, and the description complexity of $S(W)$ is bounded by $(\Delta_l+r_l^c,s_l+r_l^c)$.
	\end{enumerate}
	
	Here $c>0$ is a constant that depends only on the dimension $d$.
\end{lemma}

\begin{proof} 
The bound on $\mu(W)$, for $W\in \Lambda_t(U)$, is an immediate corollary of our definition of $\Lambda_{t}(U)$ via the terminal nodes of $\T_{t}(U)$. 
To bound the cardinality of $\Lambda_{t}(U)$, it is enough to observe that: (i) every node $\alpha$ in the refinement tree $\T_t(U)$ has only $r_l^{O(1)}$ children,  (ii) at most $c_{l}\cdot r_{l}\leq r_l^{1+\eta/10}$ of these children of $\nu$ are non-terminal, and (iii) the depth of $\T_{t}(U)$ is limited by $h_l=\log_{r_l}t$. %Since the depth of the refinement tree $\T_{t}(U)$ is limited by $\log_{r_l}t$.

Now let $W\in \Lambda_t(U)$ be a part with $l(W)>0$. If $W$ is ordinary, then we clearly have that $S(W)=S(U)$. Furthermore, the bound on the complexity of $A(W)$ stems from the fact that it is an intersection of $S(W)$ with $\log_{r_l}t$ ``prismatic" cells $\pi^{-1}_S(A_j)$, whose description complexities can be bounded by $\left(r_l^{O(d^3)},r_l^{O(d^4)}\right)$ via Theorem \ref{Theorem:ComplexityCell}. 
\end{proof}

 %, then we add to $\R_t(U)$ the triple $(U_\nu,\sigma(U_\nu)=\X(U),\pi_\X)$ with the same intrinsic dimension $l$. The fact that $|U_\nu|\leq |U|/r_l^{\log_{r_l} t}\leq |U|/t$ readily implies that $w_{i+1}(U_\nu)\leq w_{i}(U)/t$. 
	 
% If the cardinality of the canonical set $U_\nu$ is at most $c_0$, then we have that $\X(U_\nu)=U_\nu$. %, and the same projection mapping $\pi=\pi^{i-1}_j:\reals^d\rightarrow \reals^{d'}$.
	
 %Lastly, if the projection $\overline{U}_\nu$ is contained in the zero set $Z_\xi$ of its parent node $\xi$, and falls in one of the respective $\theta$-monotone patches $Z_j$ of $Z_\xi$ in $\reals^l$, then we add to $\R_t(U)$ the subset $U_\nu$ with $\X(U_\nu)=\X_\nu:=\X(U)\cap \pi^{-1}(Z_j)$  with the intrinsic dimension  $l-1$, whose projection mapping $\pi_{\X_\nu}=\pi_{\vec{\theta}}\circ\pi:\reals^{d}\rightarrow \reals^{l-1}$  is obtained by applying the ``$\pi_\theta$-pull-back" to $\pi$. Since $Z_j$ is a $\vec{\theta}$-monotone patch of $Z_\xi$, it follows that the new $(l-1)$-dimensional projection is again injective over $\X(U_\nu)$. Notice that the weight of each vertex $v\in U_\nu$ has been updated to $w_{i+1}(v):=t^{(l-1)-d}=w_{i}(v)/t$, so that $w_{i+1}(U_\nu)\leq w_i(U)/t$.

\subsection{The analysis}\label{Subsec:Analysis}

Recall that every new partition $\Pi^i$ is obtained by applying the $t$-refinement $\Lambda_t(U)$ to all sets $U\in \Pi^{i-1}$. %Hence, the asserted bounds on the cardinality of the $i$-th partition $\Pi^i$, and on the weight $\mu(U)$ of each part $U\in \Pi^i$, follow as a direct corollary of Proposition \ref{Prop:Refinement}.
Let us first summarize the most immediate properties of our partitions $\Pi^i$, which follow as a direct corollary of Lemma \ref{Lemma:Refinement}.

\begin{lemma}\label{Lemma:TreeElementary} 
%The following statement holds for any fixed choice of parameters $r_1,\ldots,r_d$ and $t$, which satisfy $\max_{l}c_l\lll_\eta r_d\lll_\eta r_{d-1}\lll\ldots\lll_\eta r_1\lll_\eta t$.

The following statement holds with any sufficiently large choice of $t$.

Let $P$ be a finite point set in $\reals^d$ and $i\geq 0$. Then the $i$-th partition $\Pi^i=\Pi^i(P)$ is comprised of $O\left(t^{i(1+\eta/(5k))}\right)$ parts which satisfy 
$\mu(U)\leq n/t^i$. 

Moreover, let $U\in \Pi^i$ be a subset of intrinsic dimension $l=l(U)>0$, and $W\in \Lambda_t(U)$ be a ``descendant" set of intrinsic dimension $l(W)>0$ in $\Pi^{i+1}$. Then the description complexity of $S(W)$ is bounded by the pair $(\Delta_l,s_l)$ with $\max\{\Delta_l,s_l\}\leq r_{l+1}^{c}$. Furthermore, if $W$ is an ordinary part of $\Lambda_t(U)$, then it is assigned an additional semi-algebraic set $A(W)\subseteq S(W)$ whose description complexity is bounded by the pair $(\Delta'_l,s'_l)$, with $\Delta'_l\leq r_l^c$ and $s'_l\leq r_l^c\log t$.

Here $c$ is a constant that depends only on the dimension $d$.
\end{lemma}

In view of our choice $I=\log_t 1/\eps^{d+\delta'}+i_0$, it follows that the cardinality of the last partitition $\Pi(P,\eps)=\Pi^I$ satisfies $\left|\Pi^I\right|=O\left((1/\eps)^{d+\delta'+\eta}\right)=O\left(1/\eps^{d+2\delta'}\right)$.
Hence, the rest of this section is dedicated to the analysis of the ``irregularity" values $\rho(\Pi^i,\psi)$ of the partitions 
 $\Pi^i$ with respect to $k$-wise relations $\psi\in \Psi_{d,k,\Delta,s}$.
According to Lemma \ref{Lemma:Refinement}, any particular $k$-family $\U=(U_1,\ldots,U_k)\in \F^i_k$ can give rise to $O\left(t^{k(1+\eta/(10k))}\right)$ $k$-families $\W\in \F^{i+1}_k$, so that $W_j\subseteq U_j$ would hold for all $1\leq j\leq k$.  Lemma \ref{Lemma:KeyLemma} below implies that only few among such $k$-families can be ``$\psi$-irregular", for any given $k$-wise relation $\psi\in \Psi_{d,k,\Delta,s}$.

%The {\it co-dimension} of a $k$-family $\U=(U_1,\ldots,U_k)\in \FF^i(V)$ is defined as the number $d-l$, where $l$ is the {\it smallest} intrinsic dimension that is attained by the member sets $U_i$ of $\U$.

%\medskip
%The following lemma implies that the overwhelming majority of the $k$-families $\W$ that descend from a given $k$-family $\U\in \I^i(V,\psi)$, are $\psi$-homogeneous.

\bigskip
\noindent{\bf Definition.} Let $\U=(U_1,\ldots,U_k)$ be an $k$-family in $\reals^d$. We say that a $k$-family $\W=(W_1,\ldots,W_k)$ {\it derives} from $\U$ if each of its member sets $W_j$, for $1\leq j\leq k$, satisfies $W_j\subseteq U_j$, and has the same intrinsic dimension $l(W_j)=l(U_j)$.%\footnote{Let us emphasize that, in addition to $U_j$, $S_j$, $\pi_j$ and $t$, the $t$-refinements $\Lambda_t(U_j)$ are also affected by the underlying constants $r_j$ and $n_0$, whose choice depends $d,k,\Delta,s$ and $\eta>0$.}

\begin{lemma}[``The key lemma"]\label{Lemma:KeyLemma} For any choice of $d\geq 1$, $k\geq 0$, $\Delta\geq 0$, $s\geq 1$, $c\geq 1$, and $\eta>0$,
one can fix the constants $r_l$ in Section \ref{Subsec:Refinement} so that the following property would hold for all the $m$-wise relations $\psi\in \Psi_{d,m,\Delta,s}$, with $0\leq m\leq k$:
 
	Let $\U=(U_1,\ldots,U_m)$ be an $m$-family of sets in $\reals^d$, where each set has positive intrinsic dimension $l_j=l(U_j)$, and is assigned a subset $S_j=S(U_j)$ whose description complexity is at most $(r_{l_{j}+1}^c,r_{l_j+1}^c)$, and a projection function $\pi_j=\pi_{S_j}:\reals^{d}\rightarrow \reals^{l_j}$ that is injective over $S_j$.
	Then all but $O\left(t^{m-1/d+\eta/5}\right)$ among $m$-families $\W=(W_1,\ldots,W_m)$, that derive from $\U$, are $\psi$-homogeneous.\footnote{Note that if the $k$-family $\U\in \F^i_k$ is $\psi$-homogeneous, then so is any $k$-family $\W=(W_1,\ldots,W_k)$ that ``descends" from $\U$ in the subsequent families $\F^{i'}_k$, for $i'>i$. Hence, it can be assumed with no loss of generality that the family $\U$ in the statement of Lemma \ref{Lemma:KeyLemma} is {\it not} $\psi$-homogeneous.}% where the implicit constants of proportionality do not depend on $t$.

	\end{lemma}
	
	We postpone the somewhat technical proof of Lemma \ref{Lemma:KeyLemma} to Section \ref{Subsec:ProofKey}, and proceed with the ``$\psi$-irregularity" analysis of the partitions $\Pi^i(P)$.

\begin{corollary}\label{Corol:Key}
Let $d>0,k\geq 1$,$\Delta\geq 0$ and $s\geq 0$ be integers, and $c$ be the constant in Lemma \ref{Lemma:TreeElementary}. With the constants $r_l$ %and $t=t(d,k,\Delta,s,\eta)$ 
selected according to Lemma \ref{Lemma:KeyLemma}, we have that 
$$
\rho\left(\Pi^{i+1},\psi\right)=O\left(\frac{\rho\left(\Pi^{i},\psi\right)}{t^{\frac{1}{d}-\eta/5}}+\frac{n^k}{t^{(i+1)(1-\eta/5)-1}}\right)
$$
\noindent for all sets $P$ in $\reals^d$, all relations $\psi\in \Psi_{d,k,\Delta,s}$, and all $i\geq 0$. 
\end{corollary}
	
\begin{proof} Fix a relation $\psi\in \Psi_{d,k,\Delta,s}$, and note that every ``$\psi$-irregular" $k$-family $\W=(W_1,\ldots,W_k)\in \G^{i+1}_k(P,\psi)$ falls into one of the following categories:

\begin{enumerate}
  
	\item It derives from such a $k$-family $\U\in \G^{i}_k(P,\psi)$.
	\item It ``descends" from some $k$-family $\U=(U_1,\ldots,U_k)\in \G_k^{i}(P,\psi)$, so that $W_j\in \Lambda_t(U_j)$ holds for all $1\leq j\leq k$. However, $\W$ includes a set $W_j\in \Lambda_t(U)$ so that $l(W_j)<l(U_j)$. 
	%\item  At least one of its member sets $W_j$ belongs to the previous partition $\Pi^{i}$, where it is assigned intrinsic dimension $l(W_j)=0$ (as $|W_j|\leq n_0$).
	\item Some pair of its sets $W_{j'},W_{j''}$, with $1\leq j'<j''\leq k$, belong to the refinement $\Lambda_t(U_{j})$ of the same set $U\in \Pi^{i}$. 
\end{enumerate}

\noindent{\bf Case 1.} Note, first, that the contribution of the $k$-families $\W$ of type 1 to $\rho(\Pi^{i},\psi)$ does not exceed $O\left(t^{-1/d+\eta/5}\cdot \rho(\Pi^{i},\psi)\right)$. 
Indeed, according to Lemma \ref{Lemma:KeyLemma}, 
only $O(t^{k-1/d+\eta/5})$ ``$\psi$-irregular" $k$-families $\W=(W_1,\ldots,W_k)\in \G_k^{i+1}(P,\psi)$ can derive from any particular family $\U=(U_1,\ldots,U_k)$ in $\G_k^i(P,\psi)$ (whose contribution to $\rho(\Pi^i,\psi)$ was $\prod_{j=1}^k \mu(U_j)$). Using that $\mu(W_j)\leq \mu(U_j)/t$ holds for all $1\leq j\leq k$, it follows that the overall contribution of such $k$-families $\W$ to the quantity $\rho(\Pi^{i+1},\psi)$ does not exceed
$$
O\left(\frac{t^{k-\frac{1}{d}+\eta/5}}{t^k}\cdot \prod_{j=1}^k \mu_{i}(U_j)\right)=O\left(t^{-\frac{1}{d}+\eta/5}\prod_{j=1}^k \mu_i(U_j)\right),
$$

\noindent for a total weight of $\displaystyle O\left(\frac{\rho(\Pi^i,\psi)}{t^{1/d-\eta/5}}\right)$.

\medskip
\noindent{\bf Case 2.} To bound the contribution of the ``$\psi$-irregular" $k$-families of type 2, let $\Lambda'_t(U)$ denote the collection of all the special subsets in $\Lambda_t(U)$, for any $U\in \Pi^{i}$. (Recall that these are the subsets $W\in \Lambda_t(U)$ with $l(W)<l(U)$.)
Since (i) every subset $W\in \Lambda_t(U_j)$ satisfies $\mu_{i+1}(W)\leq \mu_{i}(W)$ and, moreover, (ii) every special subset $W\in \Lambda_t'(U_b)$ satisfies $\mu_{i+1}(W)\leq \mu_{i}(W)/t$, the overall ``contribution" of such $k$-families $\W\in \F_k^{i+1}$ that descend from a given $k$-family $\U=(U_1,\ldots,U_k)\in \F^{i}_k$, does not exceed

$$
\sum_{j=1}^k\left( \sum_{W\in \Lambda'_t(U_j)}\mu_{i+1}(W)\right)\left(\prod_{j'\in [k]\setminus \{j\}}\mu_{i}(U_{j'})\right)\leq \sum_{j=1}^k\frac{\mu_{i}(U_j)}{t}\left(\prod_{j'\in [k]\setminus \{j\}}\mu_{i}\left(U_{j'}\right)\right)=k\cdot \frac{\prod_{j=1}^k\mu_i\left(U_j\right)}{t}.
$$

\noindent Repeating this argument for all non-homogeneous $k$-families $\U\in \G_{k}^{i}$ shows that the overall ``contribution" of the families $\W\in \F^{i+1}_k$ of type 2 to the quantity $\rho(\Pi^{i+1},\psi)$ is only $O\left(\rho\left(\Pi^{i},\psi\right)/t\right)$.

%\medskip
%\noindent{\bf Case 3.} In view of the bound in Lemma \ref{Lemma:TreeElementary}, the overall contribution of the $k$-families $\W=(W_1,\ldots,W_k)$ of type 3  amounts to at most $O(|\Pi^{i}|\cdot n_0\cdot n^{k-1})=O\left(t^{i(1+\eta/(5k))}\cdot n^{k-1}\right)$.

\medskip
\noindent{\bf Case 3.} As can be easily deduced via Lemmas  \ref{Lemma:Refinement} and \ref{Lemma:TreeElementary}, the number of the $k$-families of type 3 is $O\left(|\Pi^{i+1}|^{k-1}\cdot t^{1+\eta/(10k)}\right)=O\left(t^{(i+1)(k-1)(1+\eta/(5k))+1+\eta/(10k)}\right)$, so their contribution to $\rho(\Pi^{i+1},\psi)$ is only 
$$
O\left((n/t^{i+1})^k\cdot t^{(i+1)(k-1)(1+\eta/(5k))+1+\eta/(10k)}\right)=O\left(n^k/t^{(i+1)(1-\eta/5)-1}\right).
$$
\end{proof}

\noindent{\bf Wrap-up.} 
%If $|P|=O\left(1/\eps^{d+\delta}\right)$, then our regular partition is comprised of $n$ singletons, which guarantees that {\it all} $k$-families are homogeneous. Hence, it can be assumed in what follows that $n\geq C/\eps^{d+\delta}$, with a suitably large constant $C>0$.
%It can be assumed that the inequality
%\begin{equation}\label{Eq:BigSet}
	%|P|=n\geq \max\{c/\eps^{d+1+\delta},n_0/\eps\}
%\end{equation}
%\noindent holds with a sufficiently large constant $c$.
Recall that we have $I=\log_t\left(1/\eps^{d+\delta'}\right)+i_0$, with some small integer constant $i_0>0$, and use induction to show that, given suitably large choice of $t=t(d,k,\Delta,s,\delta)$ and $c$, the inequality
\begin{equation}\label{Eq:Induction}
\rho(\Pi^i,\psi)\leq n^kt^{2-\frac{(i-1)(1-\delta')}{d}}.
\end{equation}
must hold for all  $0\leq i\leq I$, and all $k$-wise relations $\psi\in \Psi_{d,k,\Delta,s}$. Choosing a sufficiently large constant $i_0>0$ will then guarantee that the last partition is almost $(\eps,\psi)$-regular, for any $\psi\in \Psi_{d,k,\Delta,s}$.

Since the statement of the induction is trivial for $i=0$ (as $\Pi^0=\{P\}$), let us fix $1\leq i\leq I$, and suppose that (\ref{Eq:Induction}) holds for all the previous partitions $\Pi^{i'}$, with $i'<i$, and all $k$-wise relations $\psi\in \Psi_{d,k,\Delta,s}$. To establish the same claim for $\Pi^i$, and a particular relation $\psi\in \Psi_{d,k,\Delta,s}$, we use Corollary \ref{Corol:Key} to derive
$$
\rho\left(\Pi^{i},\psi\right)\leq c'_1\cdot \frac{\rho\left(\Pi^{i-1},\psi\right)}{t^{\frac{1}{d}-\eta/5}}%c'_2\cdot t^{(i-1)(1+\eta/(5k))}n^{k-1}
+c'_2\cdot \frac{n^k}{t^{i(1-\eta/5)-1}},
$$

\noindent where the constants $c'_1$ and $c'_2$ do not depend on $t$ or $i_0$. Plugging the induction assumption into the first term yields

$$
\rho\left(\Pi^{i},\psi\right)\leq c'_1\cdot n^{d}\cdot {t^{-\frac{1}{d}+\eta/5}\cdot t^{2-\frac{(i-2)(1-\delta')}{d}}}+c'_2\cdot \frac{n^k}{t^{i(1-\eta/5)-1}}.
$$

Recall that we have $\eta=\delta'/(100dk)$. %In view of the lower bound (\ref{Eq:BigSet}), and that $i\leq I=\log_t\left(1/\eps^{d+\delta'}\right)+i_0$, the term on the right hand side can be bounded by $c''_2\eps^{(d+\delta/2)} n^{k-1}\leq c''_2n^k/t^{(i-1)(1-\delta')/d}$. 
Hence, choosing suitably large constant $t$, which may depend on $c'_1,c'_2$, $i_0$ and $\delta'$, guarantees that the right hand side does not exceed $n^k\cdot t^{2-\frac{(i-1)(1-\delta')}{d}}$ for all $0\leq i\leq I$. 

Finally, fixing a suitable integer $i_0>d$ (and recalling that $\delta'=\delta/(100dk)$) yields
$$
\rho(\Pi^I,\psi)\leq n^k\cdot t^{2-\frac{(I-1)(1-\delta')}{d}}\leq \eps k!{n\choose k},
$$

\noindent and that $|U|\leq \max\{n/t^{I-d},1\}\leq \eps n$ holds for each set $U\in \Pi^I$. Combined with the previous bound $O\left(1/\eps^{d+2\delta'}\right)=O\left(1/\eps^{d+\delta}\right)$ on the cardinality of $\Pi^I(P)=\Pi(P,\eps)$, and the bound $\max\{n/\eps^{d+\delta'},1\}\leq \eps n$ on the cardinalities of its member sets $U$ (both of them using Lemma \ref{Lemma:TreeElementary}), this completes the proof of Theorem \ref{Theorem:NewRegularitySharp}. $\Box$

\subsection{Proof of Lemma \ref{Lemma:KeyLemma}}\label{Subsec:ProofKey}

%Since the claim is trivial if $k=0$ (so that $\chi$ is attain a uniform value $0$ or $1$), let us assume in what follows that $k\geq 1$, and that we have established the claim for all $k'$-uniform instances, with $0\leq k'<k$.
For the sake of the $t$-refinement of Section \ref{Subsec:Refinement}, let us fix the integers $r_1,\ldots,r_d\geq 0$ that satisfy $\max\{\Delta,s\}\lll_\eta r_d\lll_\eta r_{d-1}\lll_\eta\ldots\lll_\eta r_1$, whose selection will be explained in the sequel,  the parameter $t\ggg_\eta r_1$, and the $k$-family $\U=(U_1,\ldots,U_k)$ in $\reals^d$. 

Let $\W=(W_1,\ldots,W_m)$ be a $k$-family that derives from $\U$. Recall that every set $W_j$ has intrinsic dimension $l_j=l(U_j)=l(W_j)$, belongs to the $t$-refinement $\Lambda_t(U_j)$, and is equipped with a semi-algebraic sets $S(W_j)$ and $A(W_j)\subseteq S(W_j)$. 
If $l_j=l(W_j)=0$, then we have that $|W_j|=|U_j|=1$ and, accordingly, define $S(W_j)=A(W_j):=W_j$, so that the description complexity of both $W_j$ and $A(W_j)$ is trivially bounded by $(\Delta_0,s_0)=(1,d)$.
Otherwise, Lemma \ref{Lemma:Refinement} yields that the respective description complexities of $S(W_j)$ and $A(W_j)$ are bounded by the pairs $(\Delta_{l_j},s_{l_j})$ and $(\Delta'_{l_j},s'_{l_j})$ which satisfy $\max\{\Delta_{l_j},s_{l_j}\}\leq r_{l_j+1}^c$, $\Delta'_{l_j}\leq r_{l_j}^c$ and $s'_{l_j}\leq r_{l_j}^{c}\log t$. Here $c$ denotes a suitable constant that depends only on the dimension $d$.\footnote{To simplify the exposition, we identify this constant with the constant $c$ in the lemma hypothesis.}

We say that such an $m$-family $\W$ is {\it $\psi$-separated} if the value of $\psi(a_1,\ldots,a_m)$ is invariant in the choice of $(a_1,\ldots,a_m)\in A(W_1)\times\ldots\times A(W_m)$; otherwise, we say that $\W$ is {\it $\psi$-crossed}.\footnote{The latter case occurs if and only if $A:=A(W_1)\times\ldots\times A(W_k)$ is crossed by the semi-algebraic set $Y:=\{(y_1,\ldots,y_k)\in \reals^{d\times k}\mid \psi(y_1,\ldots,y_k)=1\}$, in the sense that each of the regions $Y$ and $\reals^d\setminus Y$ contains a point of $A$.} 
Note that any $\psi$-separated $m$-family $\W$ that descends from $\U$ is, in particular, $\psi$-homogeneous. Hence, it suffices to show that, with a suitable choice of the constants $r_j=r_j(d,k,\Delta,s,\eta,c)$ for $1\leq j\leq d$, only $O\left(t^{m-1/d+m\eta/(10k)}\right)$ $\psi$-crossed $m$-families $\W$ can derive from $\U$.

The proof uses ``downward" induction the number $\ell(\U)$ which is defined as the {\it smallest} among the intrinsic dimensions $l_b=l(U_b)$ of the member sets $U_b\in \U$. For the sake of our induction, we also allow ``$0$-families" $\U=(\cdot)$ with $\ell(\U):=d+1$ and $r_{d+1}=0$. 
Notice that if $\ell(U)=l$, then the $t$-refinements $\Lambda_t(U_b)$, for $1\leq b\leq k$, do not depend on the choice of the constants $r_j$, for $1\leq j<l$.

This induction proceeds through $d+1$ steps $0\leq j\leq d$.  At the $j$-th step, we assume that the constants $r_{d+1}=0$ and $r_d\lll_\eta\ldots\lll_\eta r_{d-j+1}$ at the previous steps $0,\ldots,j-1$, have be pre-selected so as to ensure that the induction statement holds for all the $k$-families $\U$ that satisfy $d-j+1\leq \ell(\U)\leq d+1$, and proceed to establish the claim for the $k$-families $\U$ that satisfy $\ell(\U)=d-j$. Note that the assumption is trivial (and superfluous) in the $0$-th step, where all the sets $U_b\in \U$ must have the same intrinsic dimension $l_b=d$.
%Furthermore, the claim is immediate if $\ell(\U)=0$, in which case no $k$-families $\W$ can descend from $\U$. Hence, it can be assumed in what follows that $0\leq j\leq d-1$. 

For the sake of brevity, let us denote $l:=d-j$.
To proceed, let us fix the $k$-family $\U=(U_1,\ldots,U_k)$ with $\ell(\U)=l$ and assume, with no loss of generality, that the first $q$ sets $U_1,\ldots,U_q$, for some $1\leq q=q(\U)\leq m$, have intrinsic dimensions $l(U_1)=l(U_2)=\ldots=l(U_q)=l$, whereas the last $k-q$ sets $U_b$ satisfy $l(U_{b})\geq l+1$.
% If $k=0$, then the family $\FF_k^{i+1}$ is empty. %, and any $k$-wise relation $\psi\in \Psi_{d,0,\Delta,s}$ is a constant, so that $\I^i(V,\psi)=\I^{i+1}(V,\psi)=\emptyset$. 
% Assume, then, that $k\geq 1$, and that the claim has been established for all the $(k-1)$-wise relations in $\Psi_{d,k-1,\Delta,s}$, and all the $(k-1)$-families $\U'=(U_1,\ldots,U_{k-1})\in \FF^i_k$ of co-dimension $d-l$ or smaller. To establish the same statement for the $k$-wise relations $\psi\in \Psi_{d,k,\Delta,s}$, let us fix $0\leq i\leq p-1$, and a $k$-family $\U=(U_1,\ldots,U_k)\in \FF^i_k$ of co-dimension $d-l$. 

% \medskip
%To proceed, we distinguish between several types of $k$-families $\W=(W_1,\ldots,W_k)\in \I^i(V,\psi)$ that descend from $\U$.

% \medskip
%\noindent{\bf Case 1.} We have that $W_i=U_i$ for some $1\leq i\leq k$.
%In view of Proposition \ref{Prop:TreeElementary}, there can be only $O(t^{k-1-\delta})$ such $k$-families in $\I^i(V,\psi)$.
%Therefore, it can be assumed, in what follows, that every member set $W_i$ of $\W$ belongs to the $t$-refinement $\R_t(U_i)$ of the respective set $U_i$.

% \medskip
%\noindent{\bf Case 2.} At least one of the member sets $W_i$ of $\W$ is a special set of $\R_t(U_i)$.

To facilitate the subsequent analysis, for any $q$-family $(W_{1},\ldots,W_q)$ that derives from $(U_{1},\ldots,U_q)$, we fix a ``representative" $q$-tuple $(w_{1},\ldots,w_q)\in W_{1}\times\ldots\times W_q$.
Then any $\psi$-crossed $m$-family $\W=(W_1,\ldots,W_m)$ must fall into exactly one of the following two cases:
 
 \begin{itemize}
 	\item[(a)] There exist $(m-q)$-tuples $(a_{1},\ldots,a_{m-q})\in A(W_{q+1})\times\ldots\times A(W_m)$ and $(a'_{1},\ldots,a'_{,-q})\in A(W_{q+1})\times\ldots\times A(W_m)$ so that $\psi(w_1,\ldots,w_q,a_1,\ldots,a_{m-q})\neq \psi(w_1,\ldots,w_q,a'_1,\ldots,a'_{m-q})$.
 	\item[(b)] All the $(k-q)$-tuples $(a_{1},\ldots,a_{m-q})$ in $A(W_{q+1})\times\ldots\times A(W_m)$ attain the same value, say $\psi(w_1,\ldots,w_q,a_1,\ldots,a_{m-q})=1$.
 \end{itemize}
 
Note that the first scenario cannot occur for $q=m$, whereas the second one can arise only for positive values $\ell(\U)=l$.
 In the sequel, we establish a separate bound on the number of the $\psi$-crossed $m$-families $\W$ that fall into each category.
 
 \medskip
\noindent{\bf Case (a).} To bound the number of the $\psi$-crossed $m$-families $\W$ of the first kind, 
let us fix a $q$-family $\W'=(W_{1},\ldots,W_q)$, that derives from $\U'=(U_{1},\ldots,U_q)$, and establish an upper bound on the maximum number of such $(m-q)$-families $\W''=(W_{q+1},\ldots,W_m)$ that derive from $\U''=(U_{q+1},\ldots,U_m)$ and ``complement" $\W'$ to a $\psi$-crossed $k$-family $\W=\W'\ast \W''=(W_1,\ldots,W_m)$ of type (a).

\medskip
\noindent{\it The relation $\psi''$.} Notice that the prior choice of the representatives $(w_{1},\ldots,w_q)\in W_{1}\times\ldots\times W_q$ determines a $(m-q)$-wise relation ${\psi}''\in \Psi_{d,m-q,\Delta,s}
$ with description $\left(g''_1,\ldots,g''_s,\phi\right)$, where each function $g''_i$ is obtained by ``fixing" the first $q$ arguments in $g_i$ to $w_{1},\ldots,w_q$. The crucial observation is that every $(k-q)$-family $\W''$ in question must be $\psi''$-crossed. Therefore, and since we have $\ell(\U'')\geq l+1$, applying the induction assumption to the $(m-q)$-family $\U''$ yields an upper bound $O\left(t^{m-q-1/d+(m-q)\eta/(10k)}\right)$ on the number of such $\psi''$-crossed $(m-q)$-families $\W''$ that derive from $\U''$. 

Repeating the above argument for all $1\leq q\leq m$, and all the $O(t^{q+q\eta/(10k)})$ possible choices of the $q$-family $\U'=(U_1,\ldots,U_q)$, yields the asserted bound $O\left(t^{m-1/d+m\eta/(10k)}\right)$ on the overall number of the $\psi$-crossed $m$-families $\W=(W_1,\ldots,W_m)$ of type (a) that derive from $\U$. 

\medskip
\noindent{\bf Case (b).} Let us now fix an $(m-q)$-family $\W''=(W_{q+1},\ldots,W_m)$ that derives from $\U''=(U_{q+1},\ldots,U_m)$, and argue that, with a suitable choice of $r_l$, which satisfies $r_{l+1}\lll_\eta r_l$, it can be ``completed" in only $O\left(t^{q-1/l+q\eta/(10k)}\right)$ ways, to a $\psi$-crossed $m$-family $\W\in (W_1,\ldots,W_m)$ of type (b).

\medskip
\noindent{\it The relation $\psi'$.} To this end, let $\psi':\reals^{d\times q}\rightarrow \{0,1\}$ denote the $q$-wise relation which satisfies $\psi'(y_{1},\ldots,y_{q})=1$ if and only if $\psi(y_{1},\ldots,y_{q},a_1,\ldots,a_{m-q})=1$ holds for all the $(m-q)$-tuples $(a_1,\ldots,a_{m-q})\in A(W_{q+1})\times\ldots\times A(W_m)$.

Similar to case (a), any $\psi$-crossed $k$-family $\W=(W_1,\ldots,W_m)$ under consideration corresponds to a $\psi'$-crossed {\it $q$-family} $\W'=(W_{1},\ldots,W_q)$. Indeed, since $\W$ falls into case (b), we clearly have that $\psi'(w_{1},\ldots,w_q)=1$ (where $w_{1},\ldots,w_q$ denote the pre-selected representatives in, respectively, $W_{1},\ldots,W_q$). However, since $\W$ is $\psi$-crossed, $\psi'(y_{1},\ldots,y_q)=1$ cannot simultaneously hold for all the $q$-tuples $(y_{1},\ldots,y_q)\in A(W_{1})\times\ldots\times A(W_q)$.

Also note that the relation $\psi'$ belongs to the class $\Psi_{d,q,\Delta',s'}$, where $\Delta'=(r_{l+1}+\Delta)^{\kappa}$ and $s'=(r_{l+1}+s)^{\kappa}\log^{\kappa} t$, and $\kappa=\kappa(d,k)$ is a fixed-degree polynomial in $d$ and $k$ (i.e., the degree of $\kappa(d,k)$ does not depend on $d$).
To see this, let us observe that the set\footnote{It is instructive to note that, if $q=m$ (which is always the case for $l=d$) then $Y_{\psi'}=\{(y_1,\ldots,y_m)\in \reals^{d\times m}\mid \psi(y_1,\ldots,y_m)=1\}$.} 
$$
Y_{\psi'}:=\left\{(y_1,\ldots,y_{q})\in \reals^{d\times q}\mid \psi'(y_1,\ldots,y_{q})=1\right\},
$$
 is the complement of 
$$
\left\{(y_1,\ldots,y_{q})\in \reals^{d\times q}\mid \exists (a_1,\ldots,a_{m-q})\in  A(U_{q+1})\times\ldots\times A(U_{m}): \psi(y_1,\ldots,y_{q},a_1,\ldots,a_{m-q})=0\right\}
$$

\noindent which, in turn, is a projection of the following set within $\reals^{d\times m}$

$$
\tilde{Y}_{\psi'}:=
$$
$$
\left\{(y_1,\ldots,y_{q},a_1,\ldots,a_{m-q})\in \reals^{d\times k}\mid (a_1,\ldots,a_{m-q})\in A(U_{q+1})\times\ldots\times A(U_m), \psi(y_1,\ldots,y_{q},a_1,\ldots,a_{m-q})=0 \right\}.
$$

\noindent In view of the bound $(r_{l+1}^{c},r_{l+1}^{c}\log t)$ on the description complexity of the sets $A(U_{q+1}),\ldots,A(U_m)$, whose respective intrinsic dimensions $l(A_b)$ are all {\it strictly larger} than $l$, the description complexity of $\tilde{Y}_{\psi'}$ in $\reals^{d\times k}$ is bounded by $\left(\Delta+mr_{l+1}^c,s+mr_{l+1}^c\log t\right)$. The asserted bound on the description complexity of $Y_{\psi'}$ (which is also the description complexity of the relation $\psi'$) now follows via Theorem \ref{Theorem:Elimination}.

\bigskip
\noindent{\it The set $Y^*_{\psi'}\subseteq \reals^{l\times q}$.} For the sake of brevity, let $S_1,\ldots,S_{q}$ respectively denote the semi-algebraic sets $S(U_{1}),\ldots,S(U_{q})$, and $\pi_b$ denote the projection function $\pi_{S_b}:\reals^d\rightarrow \reals^l$ that is associated with each set $S_b$. In what follows, we use $A^*(W_b)$ to denote the projection $\pi_b(A(W_b))$ for all $1\leq b\leq q$.

Consider the projection function $\pi^*:\reals^{d\times q}\rightarrow \reals^{l\times q}$ that sends every $q$-tuple $(y_1,\ldots,y_{q})\in \reals^{d\times q}$ to the point $(\pi_1(y_1),\ldots,\pi_{q}(y_{q}))$, and denote 
$$
Y^*_{\psi'}:=\pi^*(Y_{\psi'}\cap \left(S_1\times\ldots\times S_{q}\right))\subseteq \reals^{l\times q}.
$$ 
Using that the description complexity of each set $S_b=S(W_b)$ is bounded by the pair $(\Delta_{l},s_l)$ that satisfies $\max\{\Delta_l,s_l\}\leq r_{l+1}^c$, another application of Theorem \ref{Theorem:Elimination} yields that 
the description complexity of $Y^*_{\psi'}$ is bounded by the pair $(\Delta^*,s^*)$, with $\Delta^*=(r_{l+1}+\Delta)^{\kappa^*}$ and $s^*=(r_{l+1}+s)^{\kappa^*}\log^{\kappa^*}t$, and where $\kappa^*=\kappa^*(d,m)$ is yet another constant that is polynomial in $m$ and $d$. 

The following property stems from the fact that each projection $\pi_b:\reals^{d}\rightarrow \reals^l$ is injective over the set $S_b\supseteq A(W_b)$, for $1\leq b\leq q$.

\begin{proposition}
A $q$-family $(W_1,\ldots,W_q)$, that derives from $(U_1,\ldots,U_q)$, is $\psi'$-crossed if and only if the set $A^*:=A^*(W_1)\times\ldots\times A^*(W_q)$ is crossed by $Y^*_{\psi'}$, in the sense that $A^*\cap X^*_{\psi'}\neq \emptyset$ and $A^*\not\subset Y^*_{\psi'}$.
\end{proposition}

To proceed, let us
consider the $q$ auxiliary trees $\T_t(U_1),\ldots,\T_t(U_q)$ that have been used to obtain the respective $t$-refinements $\Lambda_t(U_{1}),\ldots,\Lambda_t(U_q)$ of $U_{1},\ldots,U_q$.\footnote{Recall that case (b) does not arise if $\ell(\U)=\ell(\W)=0$, so each set $W_b$ indeed belongs to the $t$-refinement $\Lambda_t(U_b)$.} We say that a non-root node $\alpha$ in $\T_t(U_b)$ is {\it ordinary} if its canonical subset $W_\alpha$ contains at least one ordinary set $W\in \Lambda_t(U)$; note that $W_\alpha$ lies in a unique cell $C=C_\alpha$ of the partition $\reals^l\setminus Z(g_{\alpha'})$, which has been constructed for the parent node $\alpha'$ of $\alpha$ in $\T_t(U_b)$.

For any integer $b$, that satisfies $1\leq b\leq q$, and any level $1\leq h\leq h_l=\log_{r_l}t$, let us use $\N_h(U_b)$ to denote the subset of all the ordinary $h$-level nodes in $\T_t(U_b)$. In what follows, we consider the family $\M_h(U_1,\ldots,U_q):=\N_h(U_1)\times\ldots\times \N_h(U_q)$, which is comprised of all the $q$-sequences $(\alpha_1,\ldots,\alpha_q)$ of ordinary $h$-level nodes which are respectively chosen in the trees $\T_t(U_1),\ldots,\T_t(U_q)$.
For each $(\alpha_1,\ldots,\alpha_q)\in \M_h(U_1,\ldots,U_q)$, we use $C_{\alpha_1,\ldots,\alpha_q}$ to denote the cartesian product $C_{\alpha_1}\times\ldots\times C_{\alpha_q}\subseteq \reals^{l\times q}$ of the respective cells $C_{\alpha_b}\subset \reals^l$.

%For any $\psi'$-crossed $q$-family $\W=(W_1,\ldots,W_q)$ that derives from $(U_1,\ldots,U_q)$, any $1\leq b\leq q$, and any level $0\leq j\leq h-1$, let $\alpha(W_b,j)$ denote the unique $j$-level node $\alpha\in \N_t(U_b,j)$ whose canonical subset $W_\alpha$ contains $W_b$.

 %We identify the coordinates of $\reals^{l\times q}$ with $\{x_{a,b}\mid 1\leq a\leq l,1\leq b\leq q\}$. For any $l$ polynomials $g_1,\ldots,g_q\in \reals[x_1,\ldots,x_l]$, we use $g_1\otimes g_2\otimes\ldots\otimes g_q$ to denote the polynomial $\prod_{b=1}^{q} g_b(x_{1,b},\ldots,x_{l,b})$ over $\reals^{l\times q}$, and note
%that every cell $Q\subset \reals^{l\times q}\setminus Z(g_1\otimes g_2\otimes\ldots\otimes g_q)$ is the product $Q_1\times\ldots\times Q_{q}$ of some $q$ cells $Q_1\subseteq \reals^l\setminus Z(g_1),\ldots,Q_q\subseteq \reals^l\setminus Z(g_{q})$.  
Using $\left(f^*_1,\ldots,f^*_{s^*},\Phi^*\right)$ to denote the semi-algebraic description of $Y^*_{\psi'}$, with $\deg(f_e^*)\leq \Delta^*$ for all $1\leq e\leq s^*$, we obtain the following property.

\begin{proposition}\label{Prop:Fa}
	Let $\W'=(W_1,\ldots,W_q)$ be a $\psi'$-crossed family that derives from $(U_1,\ldots,U_q)$. Then there exist a polynomial $f^*_e$, with $1\leq e\leq s^*$, and a unique $q$-sequence $(\alpha_{1}(h),\ldots,\alpha_{q}(h))\in \M_h(U_1,\ldots,U_q)$ at every level $0\leq h\leq h_l$, so that the following properties hold for all $1\leq h\leq h_l$:
	
	\begin{enumerate}
		\item We have that $W_1\subseteq W_{\alpha_{1}(h)},\ldots,W_q\subseteq W_{\alpha_{q}(h)}$.
		\item The set $C_{\alpha_1(h),\ldots,\alpha_q(h)}$ is crossed by $Z(f^*_e)$.
			\end{enumerate}
\end{proposition}

\begin{proof}
Since every set $W_b$ is assigned to a unique terminal node in the respective tree $\T_t(U_b)$,
which can be reached via a unique path,
 at every level $0\leq h\leq h_l$ there is a unique $h$-level $q$-sequence $(\alpha_{1}(h),\ldots,\alpha_{q}(h))\in \M_h(U_1,\ldots,U_q)$ that meets the first condition. 
Suppose for a contradiction that for every polynomial $f^*_e$ there is a level $1\leq h=h(e)\leq h_l$ with the property that $C_{\alpha_1(h),\ldots,\alpha_q(h)}\cap Z(f^*_e)=\emptyset$. However, as $A^*(W_1)\times\ldots\times A^*(W_q)\subseteq C_{\alpha_1(h),\ldots,\alpha_q(h)}$ holds for all $1\leq h\leq h_l$,\footnote{This is because, according to the definition of the sets $A(W_b)$ in Section \ref{Subsec:Refinement}, each of them must be contained in the ``prism" of $\pi_{b}^{-1}\left(C_{\alpha_b}(h)\right)$ over the cell $C_{\alpha_b}(h)\subset \reals^l\setminus Z_{\alpha_b(h-1)}$.} it follows that the sign of every polynomial $f^*_e$, for $1\leq e\leq s^*$, is invariant over $A^*(W_1)\times\ldots\times A^*(W_q)$. Thus, contrary to Proposition \ref{Prop:Fa}, and the choice of $(W_1,\ldots,W_q)$ as a $\psi'$-crossed $q$-family, we have that either $Y^*_{\psi'}\supseteq A^*(W_1)\times\ldots\times A^*(W_q)$ or $Y^*_{\psi'}\cap \left(A^*(W_1)\times\ldots\times A^*(W_q)\right)=\emptyset$. 
\end{proof}

For any level $1\leq h\leq h_l$, and any $1\leq e\leq s^*$, let $\M_h(U_1,\ldots,U_q;f^*_e)$ denote the subset of all such $q$-sequences $(\alpha_1(h),\ldots,\alpha_q(h))\in \M_h(\U_1,\ldots,U_q)$ 
that are reached from the root sequence $(\alpha_1(0),\ldots,\alpha_q(0))$ via a path of $h'$-level $q$-sequences $(\alpha_1(h'),\ldots,\alpha_q(h'))$, for $1\leq h'\leq h$, that all satisfy $Z(f^*_e)\cap C_{\alpha_1(h'),\ldots,\alpha_q(h')}\neq \emptyset$. In addition, we denote $\M_0(U_1,\ldots,U_q;f^*_e):=\{(\alpha_1(0),\ldots,\alpha_q(0))\}$ for all $1\leq e\leq s^*$. (As before, $\alpha_b(0)$ denotes the root node of $\T_t(U_b)$.)

In view of Proposition \ref{Prop:Fa}, every $\psi'$-crossed family $\W'=(W_1,\ldots,W_q)$, that derives from $(U_1,\ldots,U_q)$, corresponds to a unique $h_l$-level $q$-sequence of terminal nodes
$$
(\alpha_1,\ldots,\alpha_q)\in \bigcup_{1\leq e\leq s^*}\M_{h_l}(U_1,\ldots,U_q;f^*_e),
$$ 
\noindent so that $W_{\alpha_b}=W_b$ holds for all $1\leq b\leq q$. Thus, our task comes down to bounding the respective cardinalities of all the sets $\M_h(\U_1,\ldots,U_q;f^*_e)$, with $1\leq h\leq h_l$ and $0\leq e\leq s^*$.

\begin{proposition}
	 With a suitable choice of the constant $r_l\ggg_\eta r_{l+1}$ in Section \ref{Subsec:Refinement}, the following statement holds at any given level $0\leq h\leq h_l-1$, and for any $1\leq e\leq s^*$:
	 
	 Any $q$-sequence $(\alpha_1,\ldots,\alpha_q)\in \M_h(U_1,\ldots,U_q;f^*_e)$ gives rise to only $O\left(r_l^{q-1/l+\eta/(20k)}\right)$ ``$f_e$-crossed" child sequences $(\beta_1,\ldots,\beta_q)\in \M_{h+1}(U_1,\ldots,U_q;f^*_e)$ in the following level $h+1$. 
\end{proposition}

\begin{proof}
	Let us fix the sequence $(\alpha_1,\ldots,\alpha_q)\in \M_h(U_1,\ldots,U_q;f^*_e)$, and consider the $r_l$-partitioning polynomials $g_{\alpha_1},\ldots,g_{\alpha_q}\in \reals[y_1,\ldots,y_l]$ that are used to subdivide the respective canonical subsets $W_{\alpha_1}\subseteq U_1,\ldots,W_{\alpha_q}\subseteq U_q$. 
We identify the coordinates of $\reals^{l\times q}$ with $\{y_{a,b}\mid 1\leq a\leq l,1\leq b\leq q\}$, and consider the polynomial 
$$
\tilde{g}_{\alpha_1,\ldots,\alpha_q}:=\prod_{b=1}^{q} g_{\alpha_b}(y_{1,b},\ldots,y_{l,b})
$$ 

\noindent over $\reals^{l\times q}$. 

Note that $\deg(\tilde{g}_{\alpha_1,\ldots,\alpha_q})=O\left(qr_l^{1/l}\right)$,\footnote{Recall that each polynomial $g_{\alpha_b}$ has been constructed in $\reals^l$, and satisfies $\deg(g_{\alpha_b})=O\left(r_l^{1/l}\right)$, where the constant of proportionality does not depend on $r_l$ or $t$} and that every Cartesian product $C_1\times\ldots\times C_{q}$ of some $q$ cells $C_1\subseteq \reals^l\setminus Z(g_{\alpha_1}),\ldots,C_q\subseteq \reals^l\setminus Z(g_{\alpha_q})$ yields a connected cell $\tilde{C}\subset \reals^{l\times q}\setminus Z(\tilde{g}_{\alpha_1,\ldots,\alpha_q})$. Invoking Theorem \ref{Theorem:ZonePolynomial} for $\reals^{l\times q}\setminus Z(\tilde{g}_{\alpha_1,\ldots,\alpha_q})$ and the hypersurface $Z(f^*_e)$, in dimension $ql$, yields an upper bound $O\left(\Delta^*r_l^{q-1/l}\right)$ on the number of such cells $\tilde{C}\subset \reals^{l\times q}\setminus Z(\tilde{g}_{\alpha_1,\ldots,\alpha_q})$ that are crossed by $Z(f^*_e)$.
 Since every subsidiary sequence $(\beta_1,\ldots,\beta_q)\in \M_{h+1}(U_1,\ldots,U_q;f^*_e)$ corresponds to a cell $\tilde{C}=C_{\beta_1,\ldots,\beta_q}$ in $\reals^{l\times q}\setminus Z(\tilde{g}_{\alpha_1,\ldots,\alpha_q})$ that is crossed by $Z(f^*_e)$, this also bounds the number of such sequences $(\beta_1,\ldots,\beta_q)$ that descend from $(\alpha_1,\ldots,\alpha_q)$ in $\M_{h+1}(U_1,\ldots,U_q;f^*_e)$. Choosing a suitably large constant $r_l\ggg_\eta \Delta^*=(\Delta+r_{l+1})^{\kappa^*}$ then guarantees that $\Delta^*\leq r_l^{\eta/(20k)}$.
 \end{proof}

It follows that the cardinality of every set $\M_h(U_1,\ldots,U_q;f^*_e)$, for $1\leq e\leq s^*$ and $1\leq h\leq h_l$, is at most $r_l^{h(q-1/l+\eta/(20k))}$. Since $h_l=\log_{r_l} t$, the number of the $\psi'$-crossed families $\W'=(W_1,\ldots,W_q)$, that derive from $\U'=(U_1,\ldots,U_q)$, does not exceed
$$
\left|\bigcup_{e=1}^{s^*}\M_{h_l}(U_1,\ldots,U_q;f^*_e)\right|=O\left(s^*t^{q-1/l+\eta/(20k)}\right)=
$$
$$
=O\left(t^{q-1/l+\eta/(20k)}(r_{l+1}+s)^{\kappa^*}\log^{\kappa^*}t\right)=O\left(t^{q-1/l+\eta/(10k)}\right).
$$

\noindent  In particular, this bounds the number of such $q$-families $\W'$ that descend from $\U'$ and are not $\psi'$-homogeneous. Repeating this argument for all the $O\left(t^{m-q+(m-q)\eta/(10k)}\right)$ possible $(m-q)$-tuples $\W''=(W_{q+1},\ldots,W_{m})$, which determine the $q$-wise relation $\psi'$ (and all choices of $1\leq q=q(\U)\leq m$), yields an upper bound $O\left(t^{m-1/l+m\eta/(10k)}\right)=O\left(t^{m-1/d+m\eta/(10k)}\right)$ on the number of the $\psi$-crossed $m$-families $\W=(W_1,\ldots,W_m)$ of type (b) that derive from $\U=(U_1,\ldots,U_m)$. This completes the proof of Lemma \ref{Lemma:KeyLemma}. $\Box$

\section{Algorithmic aspects}\label{Section:Algorithmic}

\subsection{Proof of Theorems \ref{Theorem:Construct} and \ref{Theorem:WeaklyConstruct}}\label{Subsec:Construct}

In view of the reduction in Lemma \ref{Lemma:AlmostRegular}, it is sufficient to establish Theorem \ref{Theorem:WeaklyConstruct}, which then yields Theorem \ref{Theorem:Construct} as a direct corollary.

\medskip
\noindent{\bf The construction time.} The construction of the almost $(\eps,\Psi_{d,k,\Delta,s})$-regular partition $\Pi(P,\eps)$ in Section \ref{Sec:Main} is executed in $I=O(\log 1/\eps)$ stages $\Pi_0=\{P\},\ldots,\Pi_I=\Pi(P,\eps)$, where the $i+1$-th partition $\Pi^{i+1}$ is derived from its predecessor $\Pi^i$, by 
replacing every set $U\in \Pi^i$ of cardinality $|U|\geq n_0$ with its $t$-refinement $\Lambda_t(U)$. (Recall that $|U|\geq n_0$ holds only if $U$ has positive intrinsic dimension $l(U)$, and otherwise we have that $|U|=1$.) Hence, it suffices to show that each such $t$-refinement $\Lambda_t(P)$ can be constructed in $O(|U|)$ time.

To establish the last statement, recall that $\Lambda_t(U)$ is obtained by an {\it $l$-dimensional} variant $\T_t(U)$ of the ``bounded fan-out" tree structure of Agarwal, Matou\v{s}ek and Sharir \cite[Section 5]{AgMaSa}, which we construct only to a certain level $h_l=\log_{r_l}t$; here $l$ denotes the implicit dimension $l(U)$ of $U$ within $\Pi^i$. The branching factor of $\T_t(U)$ is bounded by the maximum number $c_lr_l$ of cells in the {\it $l$-dimensional} polynomial $r_l$-partition of Theorem \ref{Thm:PolynomialPartition}. (Her $r_l^{O(1)}$ ``lower-dimensional" terminal children that can arise to an inner node of $\T_t(U)$). 
For an inner node $\alpha$, the subdivision of $Z_\alpha=Z(g_{\alpha})$ into $\theta_\alpha$-monotone patches is computed using the algorithm of Lemma \ref{Lemma:Patches}, in time $O(1)$, where the constant of proportionality depends on the implicit dimension $l$, and the degree $\deg\left(g_\alpha\right)=O\left(r_l^{1/l}\right)$ of $g_\alpha$.
To this end, the sets $S=S(U)$ need not be explicitly maintained -- it is enough to keep the respective projection functions $\pi_S:\reals^d\rightarrow \reals^l$.
Since both $t$ and $r_l$ are constants which depend only on $d,k,\Delta,s$ and $\delta>0$, the entire refinement $\Lambda_t(U)$ is obtained in $O(|U|)$ time.

\bigskip
\noindent {\bf Finding the ``$\psi$-irregular" $k$-families.}
As was mentioned in the Introduction, merely deciding whether a particular $k$-family $\U=(U_1,\ldots,U_k)$ is $\psi$-homogeneous, is a computationally challenging task whose complexity depends on the respective cardinalities of the sets $U_j$ \cite{3SUM,Degeneracy}.
 Fortunately, most applications of semi-algebraic regularity lemmas do not require an explicit subdivision of the $k$-families $\U=(U_{i_1},\ldots,U_{i_k})$ into the $\psi$-homogeneous families, and those that are not. Instead, given a $k$-wise relation $\psi\in \Psi_{d,k,\Delta,s}$, it is sufficient to obtain a certain {\it superset} $\tilde{\G}=\tilde{\G}_\psi$ of $k$-families $(U_{i_1},\ldots,U_{i_k})$ that satisfies

$$
\sum_{(U_{i_1},\ldots,U_{i_k})\in \tilde{\G}}|U_{i_1}|\cdot\ldots \cdot |U_{i_k}|\leq \eps k!{n\choose k}
$$
\noindent and {\it includes} all such $k$-families that are not $\psi$-homogeneous. The construction of $\tilde{\G}_\psi$ will overly ``mimick" the $(\eps,\psi)$-regularity analysis in Section \ref{Subsec:Analysis} and, especially, the proof of Lemma \ref{Lemma:KeyLemma}, in which the sets $U\in \Pi^i$ are effectively replaced by potentially larger ``proxy" supersets $A(U)$ of small semi-algebraic description complexity.\footnote{Recall that the set $A(U)$ is defined for a set $U\in \Pi^i$ only if $U$ has the same intrinsic dimension as its parent set in $\Pi^{i-1}$.} 
Informally, the argument in Sections \ref{Subsec:Analysis} and \ref{Subsec:ProofKey} makes no distinction between the $\psi$-homogeneous $k$-families $\U=(U_1,\ldots,U_k)\in \F_i^k$ and the $\psi$-separated ones (i.e., such $k$-families $\U=(U_1,\ldots,U_k)$ whose more ``rounded" counterparts $(A(U_1),\ldots,A(U_k))$ are $\psi$-homogeneous). 

 As a preparation, for each part $U\in \Pi^i$ of intrinsic dimension $l$ we use the auxiliary tree structures $\T_t(U)$, the projection functions $\pi_S:\reals^d\rightarrow \reals^l$, and the $r_l$-partitioning polynomials $g_\alpha\in \reals[x_1,\ldots,x_l]$ to  precompute the semi-algebraic descriptions of the sets $S(U)$ and $A(W)$ which have been assigned in Section \ref{Subsec:Refinement} to all the ordinary sets $W\in \Lambda_t(U)$. Since $t$ and $r_l$ are constants that depend on $d,k,\Delta,s$ and $\delta$, this can be achieved in additional constant time for each $U\in \Pi^i$, using the algorithm of Theorem \ref{Theorem:ComplexityCell}. 

For $0\leq i\leq I$, the query algorithm maintains a certain superset $\tilde{\G}^i_k$ of $k$-families $\U=(U_1,\ldots,U_k)$ in $\F^i_k$, which includes all such $k$-families that are not $\psi$-homogeneous.
To this end, recall that any ``$\psi$-irregular" $k$-family $\W=(W_1,\ldots,W_k)$ in $\G^i_k$ that derives from some $k$-family $\U\in \tilde{\G}_k^i$, is in particular $\psi$-crossed, in the sense that the relation $\psi$ is {\it not} fixed over $A(W_1)\times\ldots A(W_k)$. (Here $A(W_1),\ldots,A(W_k)$ are the low-complexity semi-algebraic sets that have been defined in Section \ref{Subsec:Refinement} via the respective refinement trees $\T_t(U_1),\ldots,\T_t(U_k)$ of the ``parent" sets $U_1,\ldots,U_k$.) 
Every subsequent collection $\tilde{\G}^{i+1}_k$ will now encompass all such $\psi$-crossed $k$-families $\W=(W_1,\ldots,W_k)$ that derive from some $k$-family $\U\in \tilde{\G}_k^i$, along with two special categories of $k$-families that have been described in the proof of Corollary \ref{Corol:Key}.
 
To extend the inductive argument that culminates the proof of Theorem \ref{Theorem:Weakly}, to the potentially larger collections $\tilde{\G}_k^i$, for $0\leq i\leq I$, it will be enough to maintain that the ``$\psi$-irregularity" value in step $i$

$$
\rho_{i}:=\sum_{(U_{i_1},\ldots,U_{i_k})\in \tilde{\G}_k^i}\mu(U_{i_1})\cdot\ldots \cdot \mu(U_{i_k}),
$$

\noindent still satisfies the inequality
\begin{equation}\label{Eq:IrregularityStep}
	\rho_{i+1}\leq O\left(\frac{1}{t^{\frac{1}{d}-\eta/5}}\right)\cdot \rho_i+O\left(n^k\cdot t^{-i+\eta i}\right).
\end{equation}

For the sake of our amortized analysis, every set $U$ of intrinsic dimension $l=l(U)$ will be assigned {\it cost} $1/t^{d-l}$. Accordingly, the cost of a $k$-family $(U_1,\ldots,U_k)$ will be given by $\prod_{i=1}^k\frac{1}{t^{d-l(U_i)}}$. With nominal abuse of notation, we will use $|\tilde{\G}_k^i|$ to denote the overall cost of the all the $k$-families in $\tilde{\G}_k^i$.

Given the collection $\tilde{\G}^i_k$, its successor $\tilde{\G}_k^{i+1}$ will be determined incrementally, by including the proper counterparts of the three classes of the $k$-families $\W=(W_1,\ldots,W_k)$ that arise in the proof of Corollary \ref{Corol:Key}.

%;each of its $k$-families $\W=(W_1,\ldots,W_k)$ falls into one of the three categories described in the proof of Corollary \ref{Corol:Key}.

\begin{enumerate}
	\item A $k$-family $\W=(W_1,\ldots,W_k)$ in $\F^{i+1}_k$ of the first type if it descends from such a $k$-family $\U=(U_1,\ldots,U_k)$ in $\tilde{\G}_k^i$, and is $\psi$-crossed.
Notice that the proof of Lemma \ref{Lemma:KeyLemma} yields an algorithm for finding, in constant time, a collection of $O\left(
t^{k-1/d+\eta/5}\right)$ $k$-families\footnote{Here the constant of proportionality does not depend on $t$.} $\W=(W_1,\ldots,W_k)$ that derive from $\U$ and are $\psi$-crossed. (In particular, these include all such $k$-families $\W$ that derive from $\U$ and are not $\psi$-homogeneous.)
Indeed, one can test whether a $k$-family $(W_1,\ldots,W_k)$, that derives from $\U$, is $\psi$-crossed
by testing the emptiness of the  semi-algebraic set

$$
\left\{(a_1,\ldots,a_k,a'_1,\ldots,a'_k)\in \tilde{A}(\W)\times \tilde{A}(\W)\Biggr| \psi\left(a_1,\ldots,a_k\right)\neq \psi\left(a'_1,\ldots,a'_k\right)\right\},
$$

\noindent with $\tilde{A}(\W):=A(W_1)\times\ldots\times A(W_k)$.
 Since each set $A(W_j)$ has constant description complexity, which only depends on $d,k,\Delta,s$ and $\delta$, this can be done in $O(1)$ time using the algorithm of Theorem \ref{Theorem:ComplexityCell}. 
 
 Repeating this for each $k$-family $\U$ in $\tilde{\G}^i_k$, and each $k$-family $\W$ that derives from $\U$, one can  compute all the $|\tilde{\G}_k^i|\cdot O\left(t^{k-1/d+\eta/5}\right)$ $\psi$-crossed $k$-families of the first type in overall time $O\left(\left|\tilde{\G}_k^i\right|\right)$ (where the time bound omits constant factors that involve $t$).

\item We say that a $k$-family $\W=(W_1,\ldots,W_k)\in \F^{i+1}_k$ is of type 2 if it descends from a counterpart $\U=(U_1,\ldots,U_k)$ in $\F_k^i$, so that $W_j\subseteq U_j$ for all $1\leq j\leq k$, yet at least one of the sets $W_j$ satisfies $l(W_j)<l(U_j)$.

Since cardinality of the $t$-refinement $\Lambda_t(U)$ of each set $U\in \Pi^i$ is $O\left(t^{1+\eta/(10k)}\right)=O(1)$, all the $k$-families of type 2 can be computed from $\tilde{\G}_k^i$ in $O\left(|\tilde{\G}_{k}^i|\right)$ time.
Furthermore, their amortized cost is easily checked to be $O\left(t^{k-1+\eta/10}\right)\left|\tilde{\G}_k^i\right|$.

\item Lastly, a family $\W$ is of type 3 if some pair of its sets $W_{j},W_{j'}$, for $j\neq j'$, belong to the $t$-refinement $\Lambda_t(U)$ of the same set $U\in \Pi^i$.
Since the overall cardinality of the $i$-th partition $\Pi^i$ is $O\left(t^{i(1+\eta/(10k))}\right)$, the total cost of such families is $O\left(t^{1+i(k-1)(1+\eta/(10k))}\right)$, and they can be constructed in similar time using $\Pi^i$. 

\end{enumerate}

It, therefore, follows that the total amortized time that is spent to compute every subsequent collection $\G^{i+1}_k$ from its predecessor $\tilde{\G}^i_k$ is $ct^{k-1/d+\eta/5}\cdot |\G^i_k|+O\left(t^{1+i(k-1)(1+\eta/(10k))}\right)$, which also yields the bound

$$
\left|\tilde{\G}^{i+1}_k\right|\leq ct^{k-1/d+\eta/5}\cdot |\tilde{\G}_k^i|+O\left(t^{1+i(k-1)(1+\eta/(10k))}\right),
$$ 

\noindent on the cardinality of $\G^{i+1}_k$ and guarantees that (\ref{Eq:IrregularityStep}) indeed holds for all $1\leq i\leq I-1$.
Since we have that $I=\log(1/\eps^{d+2\delta'})+i_0$, a simple recursive analysis, similar to the one that culminates the proof of Theorem \ref{Theorem:WeaklyConstruct}, implies that the $I$-th collection $\G^I_k$ can be computed, using the intermediate partitions $\Pi^i_k$ and the auxiliary $t$-refinement trees $\Lambda_t(U)$, in overall time $O\left((1/\eps)^{dk-1+\delta}\right)$. $\Box$

\subsection{Proof of Theorem \ref{Theorem:NewDensity}} \label{Subsec:Density}

The proof of Theorems \ref{Theorem:Weakly} and \ref{Theorem:WeaklyConstruct} instantly yields the following $k$-partite formulation. 

\begin{theorem}\label{Theorem:Colored} The following statement holds for all $0<\delta\leq 1$.

Let $P_1,\ldots,P_k$ be point sets in $\reals^d$ and $0<\eps\leq 1$. Then one can construct in overall $O\left(\sum_{j=1}^k|P_j|\log (1/\eps)\right)$ time, partitions $\Pi_1,\ldots,\Pi_k$ of, respectively, $P_1,\ldots,P_k$, into $O(1/\eps^{d+\delta})$ subsets $U\in \Pi_j$ whose cardinalities satisfy $|U|\leq \max\{1,\eps^{d}|P_j|\}$, and so that the inequality 

\begin{equation}
	\sum_{(U_1,\ldots,U_k)\in \tilde{\G}_\psi}|U_1|\cdot\ldots\cdot |U_k|\leq \eps|P_1|\times\ldots\times |P_k|
\end{equation}

\noindent holds for all $\psi\in \Psi_{d,k,\Delta,s}$. Here $\tilde{\G}_\psi$ denotes a certain collection of $k$-families $(U_1,\ldots,U_k)\in \Pi_1\times\ldots\times \Pi_k$, including all such $k$-families that are not $\psi$-homogeneous. 

Furthermore, the construction yields a data structures that, given a relation $\psi\in \Psi_{d,k,\Delta,s}$, returns such a collection $\tilde{\G}_\psi$ in additional time $O\left(1/\eps^{dk-1+\delta}\right)$.
\end{theorem}

\begin{proof}[Proof sketch.]
We construct the partition $\Pi_j=\Pi(P_j,\eps')$ of Section \ref{Sec:Main} for each set $P_j$, with a slightly smaller parameter $\eps'=\Theta(\eps)$.
That is, each partition $\Pi_j$ is attained through a sequence $\Pi^j_0,\ldots,\Pi^I_j=\Pi_j$ of $I=\log_t 1/\eps^{d+\delta'}+i_0$ progressively refined subdivisions which have been described in the beginning of Section \ref{Sec:Main}. To this end, we set $\Pi^0_j=\{P_j\}$ for $|P^0_j|>n_0$, and otherwise $\Pi^0_j$ consists of at most $n_0$ singletons. (In the latter case, we have $\Pi^0_j=\ldots=\Pi^I_j=\{\{p\}\mid p\in P_j\}$.)

Given a relation $\psi\in \Psi_{d,k,\Delta,s}$, we construct the family $\tilde{\G}_\psi$ by slightly modifying the ``$\psi$-irregularity" analysis of Section \ref{Subsec:Construct}. To this end, we adopt a suitably adjusted families $\F_k^i=\Pi^i_1\times \ldots\times \Pi^i_k$ and $\tilde{\G}_k^i$; the latter family includes all the ``$\psi$-irregular" families $(U_1,\ldots,U_k)\in \tilde{\G}_{k}^i$.
More precisely, we have that $\G_k^0=\F^0_k$, and every subsequent collection $\tilde{\G}_k^{i+1}$ consists of the following two categories of $k$-families within $\F^0_k$:
 
 \begin{enumerate}
 	\item  The $\psi$-crossed families $\W=(W_1,\ldots,W_k)$ that derive from the $k$-families $\U=(U_1,\ldots,U_k)\in \tilde{\G}_k^{i}$, so that $W_1\subseteq U_1,\ldots,W_k\subseteq U_k$.
 	\item Such families $\W=(W_1,\ldots,W_k)$ that ``descend" from some $k$-family $\U\in \tilde{\G}_k^i$, so that $l(W_j)<l(U_j)$ holds for some $1\leq j\leq k$.
 \end{enumerate}

 The proof of Theorem \ref{Theorem:WeaklyConstruct} implies that in each step the quantity 
 
 $$
\rho_i=\sum_{(U_1,\ldots,U_k)\in \tilde{\G}_k^i}|U_1|\cdot\ldots\cdot|U_k|
 $$
 
 \noindent decreases with each iteration by a factor of $\Omega(t^{1/d-\eta/5})$, as long as the collection $\tilde{\G}_k^i$ is not empty. (Notice that the analysis is made {\it simpler} by the fact that no ``colorful" $k$-family $\W$ in $\tilde{\G}_k^{i+1}$ falls into the third category: no $k$-family $\W\in \tilde{\G}_k^i$ includes two sets that result from the $t$-refinement $\Lambda_t(U)$ of the same set $U$ in $\Pi^i_j$, for $1\leq j\leq k$.) Hence, a suitable choice of $i_0$ guarantees that the last value $\rho_I$ falls below $\eps |P_1|\times\ldots\times |P_k|$.
 \end{proof}

Let us now invoke Theorem \ref{Theorem:Colored} with $\eps/3$ instead of $\eps$, and without explicitly computing the collection $\tilde{\G}_\psi$. 
Notice that the remaining $k$-families $\U=(U_1,\ldots,U_k)\in \Pi_1\times\ldots\times \Pi_k\setminus \tilde{\G}_{\psi}$ encompass a total of at least $(2\eps/3) |P_1|\cdot\ldots\cdot |P_k|$ $k$-tuples $(p_1,\ldots,p_k)\in P_1\times\ldots\times P_k$.
%; here $0<c'<1$ denotes a suitably small constant to be determined in the sequel.
%Refer to the ``$k$-colored" version of the almost $(\eps',\Psi_{d,k,\Delta,s})$-regular partition in Theorem \ref{Theorem:Weakly}, where intermediate partitions $\Pi^1_j=\{P_j\},\ldots,\Pi^I_j=\Pi(P_j,\eps)$ are constructed {\it separately} for all sets $P_j$ using the algorithm of Theorem \ref{Theorem:WeaklyConstruct}, in overall $O\left(\sum_{j=1}^k|P_i|\log(1/\eps)\right)$ time. (Recall that if $|P_j|=O(1/\eps^{d+1})$, then the partition $\Pi(P_j,\eps')$ is comprised of $O(1/\eps^{d+1})$ singletons, and otherwise we have that $\Pi(P_j,\eps')=O\left(1/\eps^{d+\delta}\right)$.)
%Notice that the structure in the second half of Theorem \ref{Theorem:WeaklyConstruct} can be adapted nearly verbatim to finding a superset $\G$ of at most $\eps'\prod_{j=1}^k|\Pi(P_j,\eps')|$ ``colorful" $k$-families $(U_1,\ldots,U_k)\in \Pi(P_1,\eps')\times\ldots\times\Pi(P_k,\eps')$, including all such families $\U$ that are not $\psi$-homogeneous (still in additional time $O((1/\eps)^{dk-1+\delta})$). To this end, in each step $0\leq i\leq I$ we restrict our consideration to the `colorful" combinations $\F_k^i=\Pi^i(P_1)\times\ldots\times\Pi^i(P_k)$, and do the same for $\tilde{\G}_k^i$.
Our task comes down to finding such a $k$-family $\U=(U_1,\ldots,U_k)\in \Pi_1\times\ldots\times \Pi_k\setminus \tilde{\G}_{\psi}$ that meets the following criteria (with a suitable constant $c>0$):

\begin{enumerate}
	\item all $(p_1,\ldots,p_k)\in U_1\times\ldots\times U_k$ satisfy $\psi(p_1,\ldots,p_k)=1$,
	\item we have that $|U_j|\geq c\eps^{d+1+\delta} |P_j|$ for all $1\leq j\leq k$.
\end{enumerate}

In what follows, we refer to such $k$-families $\U$ as {\it good}.

\medskip
With a sufficiently small choice of $c>0$, the $k$-families $\U$ that violate the second condition, ``account" for at most $\eps |P_1|\cdot\ldots\cdot |P_k|/3$ among the $k$-tuples in $P_1\times\ldots\times P_k$. Hence, there must remain at least
 $(\eps/3)|P_1|\times\ldots\times |P_k|$ $k$-tuples $(p_1,\ldots,p_k)\in P_1\times\ldots\times P_k$ that come from the good $k$-families $\U
 \in \Pi_1\times\ldots\times \Pi_k\setminus \tilde{\G}_{\psi}$.

To find at least one good $k$-family $\U=(U_1,\ldots,U_k)$, we proceed through a series of random trials. In each round, we select a random sequence $(p_1,\ldots,p_k)$ of $k$ distinct points $p_j\in P$. If $\psi(p_1,\ldots,p_k)=1$, we use the auxiliary structures in the proof of Theorem \ref{Theorem:WeaklyConstruct} to determine the $k$-family $\U\in \Pi_1\times\ldots\times \Pi_k$ that contains $(p_1,\ldots,p_k)$, by ``tracing" each point $p_j$ in the respective partitions $\Pi^i_j$, for $0\leq j\leq I$, and testing if $\U$ is good. (If $|P_j|=O(1/\eps)$, then the partition $\Pi_j$ consists of singletons, and we randomly and uniformly select one of them.)

Note the $k$-family $\U=(U_1,\ldots,U_k)$ is belongs to $\tilde{\G}_\psi$ if and only if it can be reached through a sequence of ``colourful" $k$-families $\W^i=(W_1^i,\ldots,W^i_k)\in \tilde{\G}^i_k$, so that $U_j\subseteq W^i_j$ would hold for all $1\leq j\leq k$ and all $0\leq i\leq I$. Here $\tilde{\G}_k^i$ denotes the sub-collection of $\F_k^i=\Pi_1^i\times\ldots\times \Pi_k^i$ that was used in the proof of Theorem \ref{Theorem:Colored}. 
Hence, using the auxiliary semi-algebraic sets $A(W_j)$ as described in Section \ref{Subsec:Construct}, one can determine whether the sampled family $\U=(U_1,\ldots,U_k)$ belongs to $\Pi_1\times\ldots\times \Pi_k\setminus \tilde{\G}_\psi$ in additional time $O(I)=O\left(\log (1/\eps)\right)$. 

Since the algorithm is poised to succeed in expected $O(1/\eps)$ rounds upon finding a good $k$-family $\U$, it can can be implemented in expected $O\left(\sum_{j=1}^k\left(|P_j|+1/\eps\right)\log (1/\eps)\right)$ time. $\Box$

\section{Discussion}\label{Sec:Discuss}
\noindent{\bf A semi-algebraic same-type lemma.} One of the more studied semi-algebraic relations $\psi$ is determined by the {\it orientation $\chi(x_1,\ldots,x_{d+1})$} of a $(d+1)$-point sequence in $\reals^d$. Here we use the orientation function $\chi:\reals^{d\times (d+1)}\rightarrow \{+,-,0\}$, which is given by

\begin{equation*}
\chi(x_1,\ldots,x_{d+1})=\sign\:{\sf det} \begin{pmatrix}
    1 & 1 & \cdots & 1\\
	x_{1}(1) & x_{2}(1)& \cdots & x_{d+1}(1)\\
	x_{1}(2)& x_{2}(2)& \cdots & x_{d+1}(2)\\
	\vdots & \vdots & \vdots & \vdots\\
	x_{1}(d)& x_{2}(d)& \cdots & x_{d+1}(d)\\
\end{pmatrix},
\end{equation*}

\noindent to define the $(d+1)$-wise relation $\psi:\reals^{d\times(d+1)}\rightarrow \{0,1\}$ over $\reals^d$ that assumes value $1$ if and only if $\chi(x_1,\ldots,x_d)=+$.

 In 1998 B\'ar\'any and Valtr \cite{BaranyValtr} established the following Ramsey-type statement with respect to order types, which underlies recent progress on the Erd\H{o}s Happly End problem in the plane \cite{Happy2} and in 3-space \cite{HappyEndHigher}.

\begin{theorem}[Same Type Lemma \cite{BaranyValtr}]\label{Theorem:SameType}
 For any positive integers $d$ and $l$, there is $c_d(l)>0$ with the following property: {\it For any finite point sets $P_1,\ldots,P_l\subset \reals^d$, whose union is in general position, there exist pairwise disjoint subsets $U_1\subseteq P_1,\ldots,U_l\subseteq P_l$, each of cardinality
$|U_i|\geq c_d(l)|P_i|$, so that for any $(d+1)$-size sub-family $(U_{j_1},\ldots,U_{j_{d+1}})$ the value $\chi(p_1,\ldots,p_{d+1})$ is invariant in the choice of the representatives $p_i\in U_{j_i}$.} 
\end{theorem}

In other words, Theorem \ref{Theorem:SameType} yields subsets $U_i\subset P_i$ that encompass at least a fixed fraction of the points of $P_i$, for $1\leq i\leq l$, and furthermore, the {\it order-types} of the $k$-sequences $(p_1,\ldots,p_l)\in U_1\times\ldots\times U_l$ is invariant in the choice of the representatives $p_i\in U_i$. (That is, the convex hulls $\conv(U_i)$, for $1\leq i\leq l$, are in ``general position" -- no $d+1$ can be crossed by a single hyperplane.) 

A recent study by Bukh and Vasileuski \cite{BukhVasil} combined the polynomial partition of Guth and Katz \cite{GuthKatz} with the ``lifting" argument in $\reals^{dk}$ to re-establish Theorem \ref{Theorem:SameType} with $c_d(l)=\Theta_d\left(1/l^{d^2}\right)$, and the same asymptotic bound in $l$ was independently established by Rubin \cite{Rubin} in a more restricted setting $P_1=P_2=\ldots=P_l$, by the means of Matou\v{s}ek's simplicial partitions \cite{PartitionTrees}. Though the somewhat ad-hoc analysis of Bukh and Vasileuski does not immediately extend to arbitrary semi-algebraic relations, it largely underlies the proofs of the ``limited" almost-regularity statement of Theorem \ref{Theorem:WeaklySharp}, and the somewhat more intricate proof of Lemma \ref{Lemma:KeyLemma}.

An almost immediate consequence of  (the $k$-partite analogues of) Theorems \ref{Theorem:Weakly} and \ref{Theorem:WeaklySharp} is the following semi-algebraic generalization of the last lower bound on $c_d(l)$.

\begin{theorem}[Semi-Algebraic Same Type Lemma]\label{Theorem:SemiSameType}
For any positive integers $d$, $k$, $\Delta$, $s$, and any $\delta>0$, there is a constant $c=c(d,k,\Delta,s,\delta)$ with the following property:

For any $l\geq k$ sets $P_1,\ldots,P_l$ in $\reals^d$, and any semi-algebraic relation $\psi\in \Psi_{d,k,\Delta,s}$ one can choose $l$ subsets $U_1\subset P_1,\ldots,U_l\subset P_l$ whose respective cardinalities satisfy $|U_i|\geq cn/l^{(k-1)d+\delta}$, and so that all the $k$-families $(U_{i_1},\ldots,U_{i_k})$ are $\psi$-homogeneous, where $1\leq i_j\leq l$.

If the relation $\psi$ is sharp over $\bigcup_{i=1}^lP_i$, then the bound improves to $|U_i|\geq cn/l^{(k-1)d}$.
	
\end{theorem}

\begin{proof} 
Let us fix $\eps=c'/l^{k-1}$, with a suitably small constant $c$ that will be determined in the sequel.
Similar to Section \ref{Subsec:Density}, we construct a separate partition $\Pi_i$ of cardinality $r_i=\Theta\left(1/\eps^{d+\delta'}\right)$ for each set $P_i$, with $\delta'=\delta/(10k)$, so that each part $U\in \Pi_i$ encompasses at most $\max\{\eps^{d+\delta'}|P_i|,1\}$ elements of $P_i$. 
%If we have $|P_l|=O\left(1/\eps^{d+\delta'}\right)$ say, then we can set $Q_l=\{p_l\}$, with an {\it arbitrary} element $p_l\in P_l$.
The argument in Section \ref{Subsec:Density} implies the following statement which is parallel to \cite[Theorem 5.2]{Rubin} and \cite[Lemma 9]{BukhVasil}: for any  $1\leq i_1<i_2<\ldots<i_k\leq l$, the partitions $\Pi_{i_1},\ldots,\Pi_{i_k}$ yield $O\left(\eps|\Pi_{i_1}|\cdot\ldots\cdot |\Pi_{i_k}|\right)$ $k$-families $\U=(U_{i_1},\ldots,U_{i_k})\in \Pi_{i_1}\times\ldots\times\Pi_{i_k}$ that are {\it not} $\psi$-homogeneous. 

We now repeat the argument of \cite[Section 3]{BukhVasil} in the new setting. To this end, let $\Pi'_i:=\{U\in \Pi_i\mid |U|\geq |P_i|/(4r_i)\}$, and notice that $|\bigcup \Pi'_{i}|\geq 3|P_i|/4$ holds for all $1\leq i\leq l$. 
Consider the $k$-uniform hypergraph $(V',E')$ over $V':=\biguplus_{i=1}^l\Pi'_i$ whose edge set describes all the $\psi$-irregular $k$-families $\U=(U_{i_1},\ldots,U_{i_k})\in \Pi'_{i_1}\times\ldots\times\Pi'_{i_k}$, for $i_1<i_2<\ldots<i_k$, and refer to a random and uniform choice $(U_1,\ldots,U_l)\in \Pi'_1\times\ldots\times \Pi'_l$. 

    For any fixed choice of indices $1\leq i_1<i_2<\ldots<i_k\leq l$, let $Y_{i_1,\ldots,i_k}$ denote the event that {\it the edge $(U_{i_1},\ldots,U_{i_k})$ belongs to $E'$}. Notice that $Y_{i_1,\ldots,i_k}$ clearly has probability $O(\eps)=O(c'/l^{k-1})$ (where the implicit constant of proportionality does not depend on $c'$), and depends on at most $k{l\choose k-1}=O(l^{k-1})$ other events $Y_{i'_1,\ldots,i'_k}$ of this kind. Hence, choosing a suitably small $c'>0$, and invoking Lov\'asz Local Lemma (see, e.g., in \cite[Corollary 5.1.2]{AlonBook}) yields with positive probability that none of the events $Y_{i_1,\ldots,i_k}$ occurs. Furthermore, since each subset $U_i\subseteq P_i$ is chosen from $\Pi'_i$, its cardinality satisfies $|U_i|\geq |P_i|/(4r_i)=\Omega\left(\eps^{d+\delta}|P_i|\right)=\Omega\left(1/l^{d(k-1)+\delta}|P_j|\right)$.
    
    The second statement, in which the relation $\psi$ is sharp over $\bigcup_{i=1}^l P_i$, is established in a similar manner, by applying the $O(1/\eps^d)$-size almost-regular partition in Section \ref{Sec:Generic} to each set $P_i$. To this end, we again use $\eps=c'/l^{k-1}$ with a sufficiently small constant $c'>0$, and distinguish between two possible scenarios for each set $P_i$. If $|P_i|\geq c''/\eps^d$, for 
    a suitably large constant $c''>0$, then we can omit from consideration a small subset of at most $|P_i|/10$ points that are left ``unclassified" by the partitioning polynomial $g_i$ of $P_i$ (along with at most $|P_i|/4$ points that belong to ``small" sets of $\Pi_i\setminus \Pi'_i$). Otherwise, the respective partition $\Pi_i$ is comprised of $O(1/\eps^d)$ singletons. 
    
    It, therefore, suffices to check that the probability of each event $Y_{i_1,i_2,\ldots,i_k}$ remains $O(\eps)$. Indeed, let us fix $1\leq i_1<\ldots<i_k\leq l$, and suppose that exactly $k'\leq k$ among the respective sets $P_{i_j}$ have cardinality $|P_{i_j}|<c''/\eps^d$. Assume with no loss of generality, that their labels are $i_1,\ldots,i_{k'}$. By fixing some $k'$ {\it singletons} $U_{i_1}=\{p_{1}\},\ldots,U_{i_{k'}}=\{p_{k'}\}$ within the respective partitions $\Pi_{i_1},\ldots,\Pi_{i_{k'}}$, and each time referring to the restricted, $(k-k')$-wise relation $\psi(p_1,\ldots,p_{k'},x_1,\ldots,x_{k-k'})$ in the variables $x_1,\ldots,x_{k-k'}\in \reals^{d}$, we readily conclude that all but $O(\eps|\Pi_{i_{k'+1}}|\times\ldots\times |\Pi_{i_k}|)$ of the $(k-k')$-families $(U_{i_{k'+1}},\ldots,U_{i_k})\in \Pi_{i_{k'+1}}\times\ldots\times \Pi_{i_k}$ are $\psi'$-homogeneous.
    (To adapt the analysis in Section \ref{Sec:Generic} to the present $k'$-partite setting, we consider the partition of $\reals^{d(k-k')}$ that is induced by the product $\prod_{j=i_{k'+1}}^k \tilde{g}_{i_j}$, where $\tilde{g}_i$ denotes the polynomial $g_i(x_{1,i},\ldots,x_{d,i})$.) Repeating this argument for all the possible choices of $U_{i_1},\ldots,U_{i_{k'}}$ again implies that only an $O(\eps)$ fraction of the $k$-families $(U_{i_1},\ldots,U_{i_k})\in \Pi'_{i_1}\times\ldots\Pi'_{i_k}$ are not $\psi$-homogeneous. 
    \end{proof}

\noindent{\bf Recent progress.} After an earlier version of this article had been made public, the bound in Theorem \ref{Theorem:Weakly} was slightly improved by Tidor and Yu, who reported a almost regular partitions of cardinality $C/\eps^d$, with a more explicit constant of proportionality $C=O_{d,k}\left((s\Delta)^d\right)$. Specializing to sharp relations $\psi$, this overly matches the bound in Theorem \ref{Theorem:WeaklySharp}. In addition, they also presented a matching lower bound construction on the cardinality of almost-regular partitions. 

Similar to the our analysis in Sections \ref{Sec:Prelim} and \ref{Sec:Main}, Tidor and Yu combine the ``lifting" argument in $\reals^{dk}$ with the Barone-Basu bound \cite{BaroneBasu} (Theorem \ref{Theorem:ZonePolynomial}).
However, in contrast to our partition in Section \ref{Sec:Main}, which is obtained through a hierarchy of constant-size partitions $\reals^d\setminus Z(g)$ (each induced by the zero set of a {\it fixed-degree} $r$-partitioning polynomial $g$ given in Theorem \ref{Thm:PolynomialPartition}), the partition of Tidor and Yu builds on a rather more intricate algebraic machinery of ``multi-level partitions", and is not accompanied by a near-linear construction algorithm. 
Unfortunately, the shear cost of manipulating polynomials of super-constant degrees (let alone traversing an arrangement of their zero sets)
seems to pose a veritable challenge to efficient implementation of either the partition of Tidor and Yu, or even the simpler partition in Theorem \ref{Theorem:WeaklySharp}: as was previously mentioned, the best known algorithm for constructing an $r$-partitioning polynomial runs in expected time $O(rn+r^3)$ \cite{AgMaSa}.

 %In view of the {\it upper bound} $d^d/l^{d}$ that was established by Bukh and Vasileuski for the quantities $c_d(l)$ in Theorem \ref{Theorem:SameType} (and, therefore, also for Theorem \ref{Theorem:SemiSameType}), it would be interesting to see a lower bound of the form $\Omega\left(1/\eps^{d}\right)$ (or even $1/\eps^{\Theta(d)}$) for the partition size $K$ in Theorem \ref{Theorem:NewRegularity}.

\end{document}